%% file: Adaptive_Nonlin.tex
\documentclass{IEEEtran} 

\usepackage{graphics} % for pdf, bitmapped graphics files
\usepackage{epsfig} % for postscript graphics files
\usepackage{amsmath} % assumes amsmath package installed
\usepackage{amssymb}  % assumes amsmath package installed

\usepackage{algpseudocode}
\usepackage{algorithm}
\usepackage{epstopdf}
\usepackage{comment}

\usepackage{hyperref}

\usepackage{tikz,pgfplots}%

\usepackage{amsthm}

\newtheorem{assumption}{\bf Assumption}

\newtheorem{theorem}{\bf Theorem}
\newtheorem{proposition}{\bf Proposition}
\newtheorem{lemma}{\bf Lemma}

\newtheorem{remark}{\bf Remark}

\usepackage{setspace}
\usepackage{tikz}               
\usepgflibrary{arrows}

\usepackage{enumitem}

\begin{document}
\title{A robust adaptive model predictive control framework for nonlinear uncertain systems}
\author{Johannes K\"ohler$^1$, Peter K\"otting$^2$, Raffaele Soloperto$^1$, Frank Allg\"ower$^1$, Matthias A. M\"uller$^2$% <-this % stops a space
\thanks{%
This work was supported by the German Research Foundation under Grants GRK 2198/1 - 277536708, AL 316/12-2, and MU 3929/1-2 - 279734922. 
The authors thank the International Max Planck Research School for Intelligent Systems (IMPRS-IS) for supporting Raffaele Soloperto.
% and by the International Max Planck Research School for Intelligent Systems (IMPRS-IS).
}
\thanks{$^1$Johannes K\"ohler, Raffaele Soloperto,  and Frank Allg\"ower are with the Institute for Systems Theory and Automatic Control, University of Stuttgart, 70550 Stuttgart, Germany.
(email:$\{$johannes.koehler, raffaele.soloperto, frank.allgower\}@ist.uni-stuttgart.de)}
\thanks{$^2$Peter K\"otting and Matthias A. M\"uller are with the Institute of Automatic Control, Leibniz University Hannover, 30167 Hannover, Germany (e-mail:koetting@stud.uni-hannover.de,mueller@irt.uni-hannover.de)}
}

\IEEEoverridecommandlockouts
\IEEEpubid{\begin{minipage}{\textwidth}\ \\[18pt] \\ \\
         \copyright 2020 John Wiley \& Sons, Ltd.  All rights reserved.
     \end{minipage}}

\maketitle
%
%%%%%%%%%%%%%%%%%%%%%%%%%%%%%%%%%%%%%%%%%%%%%%%%%%%%%%%%%%%%%%%%%%%%%%%%%%%%%%%%
\begin{abstract}
In this paper, we present a tube-based framework for robust adaptive model predictive control (RAMPC) for nonlinear systems subject to parametric uncertainty and additive disturbances.
Set-membership estimation is used to provide accurate bounds on the parametric uncertainty, which are employed for the construction of the tube in a robust MPC scheme.
The resulting RAMPC framework ensures robust recursive feasibility and robust constraint satisfaction, while allowing for less conservative operation compared to robust MPC schemes without model/parameter adaptation.
Furthermore, by using an additional mean-squared point estimate in the objective function the framework ensures finite-gain $\mathcal{L}_2$ stability w.r.t. additive disturbances.

%-----
As a first contribution we derive suitable monotonicity and non-increasing properties on general parameter estimation algorithms and tube/set based RAMPC schemes that ensure robust recursive feasibility and robust constraint satisfaction under recursive model updates.
Then, as the main contribution of this paper, we provide similar conditions for a tube based formulation that is parametrized using an incremental Lyapunov function, a scalar contraction rate and a function bounding the uncertainty.
With this result, we can provide simple constructive designs for different RAMPC schemes with varying computational complexity and conservatism.
As a corollary, we can demonstrate that state of the art formulations for nonlinear RAMPC are a special case of the proposed framework.
We provide a numerical example that demonstrates the flexibility of the proposed framework and showcase improvements compared to state of the art approaches. 
\end{abstract}


\input{Intro}

\input{Setup}

\input{Theory}

\input{Example}

\input{Sum}

%%%%%%%%%%%%%%%%%%%%%%%%%%%%%%%%%%%%%%%%%%%%%%%%%%%%%%%%%%%%%%%%%%%%%%%%%%%%%%%%
\bibliographystyle{ieeetran}  
\bibliography{Literature}  
\clearpage 
\appendix
In Appendix~\ref{app:term}, an alternative design for the terminal ingredients (Ass.~\ref{ass:term} is presented. 
Additional details regarding the numerical example (Sec.~\ref{sec:num}) can be found in Appendix~\ref{app:example}. 
\input{Terminal_conservative}

\input{Example_details}
\end{document}

%% file: Intro.tex
%!TEX root = ./Adaptive_Nonlin.tex
%%%%%%%%%%%%%%%%%%%%%%%%%%%%%%%%%%%%%%%%%%%%%%%%%%%%%%%%%%%%%%%%%%%%%%%%%%%%%%%
\section{Introduction}
Model predictive control (MPC)~\cite{rawlings2009model} is an optimization based control strategy that can cope with complex nonlinear systems and general nonlinear constraints on state and input. %, maybe~\cite{kouvaritakis2016model}
The performance and theoretical properties of MPC schemes can be highly dependent on accurate prediction models.
In particular, neglecting possible computational limitations,  the lack of an accurate prediction model is one of the main practical challenges for MPC implementations. 
This has motivated an increasing amount of research focused on online model adaptation/learning in MPC, spanning the last two decades~\cite{mayne1993adaptive,shouche1998simultaneous,kim2004robust}, with current research focused on robust adaptive formulations~\cite{tanaskovic2014adaptive,lorenzen2019robust,Lu2019RAMPC,Koehler2019Adaptive,lopez2018adaptive,adetola2011robust,guay2015robust,gonccalves2016robust}, dual/learning formulations~\cite{mesbah2018stochastic,thangavel2018dual} and machine learning based approaches~\cite{hewing2019cautious,bradford2020stochastic,manzano2018robust,mckinnon2019learn}. 
However, all of these approaches suffer from at least one of the following shortcomings: 
\begin{enumerate}[label=\alph*)]
\item limitation to restrictive system classes, such as linear systems~\cite{tanaskovic2014adaptive,lorenzen2019robust,Lu2019RAMPC,Koehler2019Adaptive} or feedback linearizable systems~\cite{lopez2018adaptive},
\item failure to provide theoretical guarantees regarding recursive feasibility, closed-loop stability and constraint satisfaction~\cite{thangavel2018dual,hewing2019cautious,bradford2020stochastic,mckinnon2019learn}, 
\item significant increase in the computational complexity~\cite[Chap.~10.4]{guay2015robust}, \cite{thangavel2018dual,hewing2019cautious,manzano2018robust}, 
\item overly conservative formulation~\cite{adetola2011robust}, \cite[Chap.~10.5]{guay2015robust}, \cite{gonccalves2016robust,manzano2018robust}.
\end{enumerate}
In this paper we provide a novel framework for RAMPC that addresses all these challenges for a class of nonlinear uncertain systems.

\subsection*{Related work}
In general, RAMPC schemes~\cite{tanaskovic2014adaptive,lorenzen2019robust,Lu2019RAMPC,Koehler2019Adaptive,lopez2018adaptive,adetola2011robust,guay2015robust,gonccalves2016robust}
 use  set membership estimation to provide a set of non falsified parameters, which is utilized in a robust tube-based MPC approach to provide robustness w.r.t. uncertain parameters and disturbances.
One of the main differentiating factors among the above mentioned approaches is the parametrization of the parameter set and the construction of the robust tube that confines all uncertain predicted trajectories. 
These design choices highly effect the computational complexity and conservatism of the resulting scheme.
In general, guaranteeing robust recursive feasibility under online adapted models is highly dependent on the interplay of the considered parametrizations and, hence, results to be non trivial in the case of nonlinear uncertain systems. 

Approaches for general linear uncertain systems with polytopic tubes  can be found in~\cite{lorenzen2019robust,Lu2019RAMPC,Koehler2019Adaptive}. 
The special case of finite impulse response (FIR) systems and unknown constant/time-varying offsets are discussed in~\cite{tanaskovic2014adaptive} and \cite{bujarbaruah2019adaptive}, respectively.
%\cite{bujarbaruah2018adaptive}

One key challenge in the design of RAMPC schemes for nonlinear systems is the requirement of a suitable robust MPC approach that is applicable to nonlinear systems with parametric uncertainty, which is currently an active research field~\cite{DynamicTube_Lopez_19,villanueva2017robust,singh2017robust,Robust_TAC_19}. 
Recently, in~\cite{lopez2018adaptive} a RAMPC scheme for the special case of feedback linearizable nonlinear systems has been proposed using boundary layer control~\cite{DynamicTube_Lopez_19}. 
Except for this recent approach for feedback linearizable systems, existing RAMPC schemes for nonlinear systems with theoretical guarantees have been developed exclusively\footnote{%
In~\cite{wang2014adaptive} the approach in~\cite{adetola2011robust} has been extended to time-varying parameters, but only a min-max RAMPC formulation was considered, which is not computationally tractable. 
} by Martin Guay and coauthors~\cite{adetola2011robust,guay2015robust,gonccalves2016robust,dahliwal2014set,adetola2014adaptive}.

In particular, parameter update schemes have been developed that provide guaranteed bounds on the parameter error for a large class of nonlinear systems, including continuous time~\cite{adetola2011robust}, discrete time~\cite{gonccalves2016robust}, time varying~\cite{dahliwal2014set}, and  nonlinearly parametrized systems~\cite{adetola2014adaptive}, compare~\cite{guay2015robust} for a general overview. 
However, all these approaches  consider a simple robust tube approach based on Lipschitz constants similar to~\cite{marruedo2002input,pin2009robust}, which can be prohibitively conservative in many scenarios, see numerical comparisons in~\cite{Robust_TAC_19,kohler2018novel,IncremStochComparison_Mesbah_19}. 

On the other hand, in the last decade several robust MPC approaches for nonlinear systems have been developed, using: 
interval arithmetic~\cite{limon2005robust}; boundary layer control for feedback linearizable systems~\cite{DynamicTube_Lopez_19}; 
$\min$--$\max$ differential inequalities~\cite{villanueva2017robust}; control contraction metrics~\cite{singh2017robust}; and incremental Lyapunov functions~\cite{Robust_TAC_19,bayer2013discrete}. 
Employing more recent robust MPC approaches may alleviate some of the inherent limitations of state of the art RAMPC schemes. 
However, due to difficulties in the analysis, these robust MPC approaches have not yet been employed to design suitable nonlinear RAMPC schemes. As stated in\footnote{%
A possible exception to this statement may be the RAMPC scheme in the thesis~\cite{lopez2018adaptive}, that addresses similar challenges as the proposed approach, but is limited to continuous time feedback linearizable systems. 
}~\cite{lopez2018adaptive}: \textit{A tractable nonlinear AMPC approach that does
not rely on unrealistic assumptions has yet to be developed.}

More recently, machine learning inspired and data driven approaches for model identification/adaptation/refinement have gained a lot of attention, e.g. using Gaussian Processes (GPs)~\cite{hewing2019cautious,bradford2020stochastic}, (local) weighted Bayesian linear regression (wBLR)~\cite{mckinnon2019learn} and kinky inference~\cite{manzano2018robust}, to name a few.
Such approaches use different (potentially less restrictive) a priori assumptions~\cite{hewing2019cautious,bradford2020stochastic,manzano2018robust} and also perform well in some experiments~\cite{hewing2019cautious,bradford2020stochastic,mckinnon2019learn}.
However, currently the corresponding robust MPC literature that can provide theoretical guarantees using such models seems rather immature compared to classical RAMPC approaches, see~\cite{benosman2018model} for a more general discussion.

\subsection*{Contribution}
In this paper, we present a tube-based framework for robust adaptive model predictive control (RAMPC) for a class of uncertain nonlinear systems subject to additive disturbances and parametric uncertainty.
The paper contains the following contributions: (i) present a general theoretical framework for nonlinear RAMPC; (ii) provide a computationally efficient RAMPC framework by extending the robust MPC framework in~\cite{Robust_TAC_19} to allow for recursive model updates; (iii) provide explicit design procedures, which contain the approaches in~ \cite{adetola2011robust,guay2015robust,gonccalves2016robust} as a special case.

First, given the plethora and diversity of existing robust MPC formulations for nonlinear systems~\cite{DynamicTube_Lopez_19,villanueva2017robust,singh2017robust,Robust_TAC_19,marruedo2002input,pin2009robust,kohler2018novel,limon2005robust,bayer2013discrete} and parameter set updates~\cite{tanaskovic2014adaptive,lorenzen2019robust,Lu2019RAMPC,Koehler2019Adaptive,lopez2018adaptive,guay2015robust,dahliwal2014set,adetola2014adaptive}, we consider the general problem of providing conditions on the robust MPC tube propagation and set updates.  
In particular, we provide general conditions regarding overapproximation of the uncertainty, nonincreasingness of the parameter set, and monotonicity properties of the tube propagation, that ensure robust recursive feasibility and robust constraint satisfaction.

Second, as the main contribution of this paper, we provide a framework for computationally efficient nonlinear RAMPC based on the nonlinear robust MPC approach in~\cite{Robust_TAC_19}.
We consider nonlinear systems  linear in uncertain parameters subject to bounded additive disturbances with nonlinear Lipschitz continuous constraints. 
We assume that the nominal nonlinear system is exponentially incrementally stabilizable with some known contraction rate $\rho$ and incremental Lyapunov function $V_{\delta}$. 
Furthermore, we design a nonlinear (state and input dependent) function $\tilde{w}$ that bounds the difference between the nominal and uncertain system.
The tube dynamic is then defined by a scalar $s$ that depends on the contraction rate $\rho$ and the function $\tilde{w}$. This formulation allows us to compute an efficient constraint tightening, while only moderately increasing the computational demand with respect to a nominal MPC scheme.
As the main technical contribution,  given this parametrization using $\rho,\tilde{w},V_{\delta}$, we derive suitable conditions that ensure robust recursive feasibility and robust constraint satisfaction. 
Thus we extend the robust MPC framework in~\cite{Robust_TAC_19} to allow for online parameter updates and reduction in the conservatism.
Furthermore, we show that by designing the objective function based on an additional mean squared point estimate $\hat{\theta}$, we can guarantee  finite gain $\mathcal{L}_2$ stability w.r.t. disturbances, extending the stability results in~\cite{lorenzen2019robust,Koehler2019Adaptive} to nonlinear systems.

Next, given the general conditions for this RAMPC framework, we provide explicit design procedures that satisfy the posed conditions. 
In particular, we propose a moving window parameter set update yielding online shrinking finite complexity hypercubes that contain the true parameters. 
Then, we provide two simple designs regarding the function $\tilde{w}$ that ensure robustness, while allowing for a flexible trade off regarding conservatism and computational complexity. 
Furthermore, we demonstrate that by using a ball  as an "incremental Lyapunov function" $V_{\delta}$, we recover the approach in~\cite{adetola2011robust,guay2015robust,gonccalves2016robust} as a special case. 

Finally, we demonstrate the applicability and advantages of the proposed framework with a nonlinear example.

To summarize, the resulting nonlinear RAMPC framework: 
 (i) reduces conservatism online using set membership updates,
 (ii) improves performance using a mean squared cost, 
(iii) can be significantly less conservative compared to state of the art approaches, i.e. \cite{adetola2011robust,guay2015robust,gonccalves2016robust},  
(iv) provides a flexible trade off regarding conservativism and computational complexity, 
and (v) avoids computationally expensive online optimizations stemming from matrix variables~\cite{hewing2019cautious,villanueva2017robust}, piece-wise definitions~\cite{manzano2018robust}, scenario trees~\cite{thangavel2018dual}, or general $\min$--$\max$ optimization~\cite[Chap.~10.4]{guay2015robust}.
 
The paper is structured as follows: 
Section~\ref{sec:setup} discusses the general problem of RAMPC.  
Section~\ref{sec:main} presents the general theory regarding the proposed framework for nonlinear RAMPC with a corresponding design procedure. 
Section~\ref{sec:num} provides a numerical example to demonstrate the applicability and advantages of the proposed approach. 
Section~\ref{sec:sum} concludes the paper.
In Appendix~\ref{app:term}, an alternative design for the terminal ingredients is presented.  

\subsection*{Notation}
The quadratic norm with respect to a positive definite matrix $Q=Q^\top$ is denoted by $\|x\|_Q^2=x^\top Q x$ and 
the minimal and maximal eigenvalue of $Q$ are denoted by $\lambda_{\min}(Q)$ and $\lambda_{\max}(Q)$, respectively. 
The positive real numbers are $\mathbb{R}_{\geq 0} = \{ r\in\mathbb{R}|r\geq  0\}$.
The vertices of a polytopic set $\Theta$ are denoted by $\theta^i\in\text{Vert}(\Theta)$. 
By $\mathcal{K}_{\infty}$ we denote the class of functions $\alpha:\mathbb{R}_{\geq 0}\rightarrow\mathbb{R}_{\geq 0}$, which are continuous, strictly increasing, unbounded and satisfy $\alpha(0)=0$. 
Denote the unit hypercube by $\mathbb{B}_\infty:=\{\theta|~\|\theta\|_\infty\leq 1\}$ and the unit ball by $\mathbb{B}_2:=\{\theta|~\|\theta\|\leq 1\}$.  
The Minkowski sum and Pontryagin difference for two sets $\mathbb{A},\mathbb{B}\subseteq\mathbb{R}^n$ are denoted by $\mathbb{A}\oplus\mathbb{B}:=\{a+b|~a\in\mathbb{A},b\in\mathbb{B}\}$ and $\mathbb{A}\ominus\mathbb{B}:=\{c|~c+b\in\mathbb{A}\forall b\in\mathbb{B}\}$, respectively.

%% file: Setup.tex
%!TEX root = ./Adaptive_Nonlin.tex
%%%%%%%%%%%%%%%%%%%%%%%%%%%%%%%%%%%%%%%%%%%%%%%%%%%%%%%%%%%%%%%%%%%%%%%%%%%%%%%
\section{Setup and general theory}
\label{sec:setup}
In this section, we derive the general conditions for tube-based RAMPC in terms of set predictions, before presenting the proposed RAMPC framework in Section~\ref{sec:main}. 
We consider a nonlinear discrete time system of the form
\begin{align}
\label{eq:dyn}
x_{t+1}=f_{w}(x_t,u_t,d_t,\theta^*),
\end{align}
with  state $x_t\in\mathbb{R}^n$, control input $u\in\mathbb{R}^m$, disturbances $d_t\in\mathbb{D}\subset\mathbb{R}^q$, time $t\in\mathbb{N}$, and perturbed dynamics $f_{w}$ with some unknown but constant parameters $\theta^*\in\mathbb{R}^p$. 
We consider point-wise in time state and input constraints
\begin{align}
\label{eq:constraint_condition}
(x_t,u_t)\in\mathcal{Z},\quad t\geq 0,
\end{align}
with some compact nonlinear constraint set
\begin{align}
\label{eq:constraint_definition}
\mathcal{Z}=\{(x,u)\in\mathbb{R}^{n+m}|~h_j(x,u)\leq 0,~j=1,\dots,r\}.
\end{align}
Denote by $\mathcal{Z}_{x}$ the projection
of $\mathcal{Z}$ on $\mathbb{R}^n$. 
The following assumption characterizes the uncertainty in the parameters $\theta$. 
\begin{assumption}
\label{ass:param_set}
At each time step $t\in\mathbb{N}$, a parameter set $\Theta_t$ is computed satisfying $\Theta_{t+1}\subseteq\Theta_t\subseteq\Theta_0$ and $\theta^*\in\Theta_t$. 
\end{assumption}

For simplicity, the following presentation does not consider any stabilizing feedback $\kappa$, which is typically used in robust tube MPC,  since it complicates the exposition. 
The following results can, however, easily be adapted to this case. 
We consider some map $\Phi$ to predict sets $\mathbb{X}_{k|t}$, that characterizes the tube.
\begin{assumption}
\label{ass:set_propagation}
There exists a map $\Phi:2^{\mathbb{R}^n}\times\mathbb{R}^m\times2^{\mathbb{R}^q}\rightarrow2^{\mathbb{R}^n}$, such that for any $(\mathbb{X},u)\subseteq\mathcal{Z}$, $\Theta\subseteq\Theta_0$, we have
\begin{align}
\label{eq:set_inclusion}
f_w(x,u,d,\theta)\in\Phi(\mathbb{X},u,\Theta),
\end{align}
for any $x\in\mathbb{X}$, $d\in\mathbb{D}$ and $\theta\in\Theta$. 
Furthermore, the function $\Phi$ satisfies the following monotonicity property
\begin{align}
\label{eq:set_monotonicity}
\Phi(\mathbb{X}',u,\Theta')\subseteq\Phi(\mathbb{X},u,\Theta),
\end{align}
for any $(\mathbb{X}',u)\subseteq(\mathbb{X},u)\subseteq\mathcal{Z}$, and any $\Theta'\subseteq\Theta\subseteq\Theta_0$.
\end{assumption}
Given these conditions, a general tube based robust adaptive MPC scheme can be formulated with the following optimization problem using the measured state $x_t$, the parameter set $\Theta_t$ and some cost function $J_N$:
\begin{subequations}
\label{eq:set_MPC}
\begin{align}
&\min_{u_{\cdot|t},\mathbb{X}_{\cdot|t}}J_N(\mathbb{X}_{\cdot|t},u_{\cdot|t})\\
\text{s.t. }&x_t\in\mathbb{X}_{0|t},\\
\label{eq:set_MPC_inclusion}
& \mathbb{X}_{k+1|t}\supseteq\Phi(\mathbb{X}_{k|t},u_{k|t},\Theta_t),\\
\label{eq:set_MPC_tightened_constraints}
&(\mathbb{X}_{k|t},u_{k|t})\subseteq\mathcal{Z},\\
&k=0,\dots,N-1,\\
\label{eq:set_MPC_term}
&\Phi(\mathbb{X}_{N|t},u_{N|t},\Theta_t)\subseteq\mathbb{X}_{N|t},~(\mathbb{X}_{N|t},u_{N|t})\subseteq\mathcal{Z}.
\end{align}
\end{subequations}
The minimizers are denoted by $u_{\cdot|t}^*$, $\mathbb{X}_{\cdot|t}^*$. 
The corresponding closed loop input is given by $u_t=u_{0|t}^*$. 
 
The following theorem shows that this general set approach directly ensures robust recursive feasibility and robust constraint satisfaction.
\begin{theorem}
\label{thm:set_robust}
Let Assumptions \ref{ass:param_set} and \ref{ass:set_propagation} hold and suppose that Problem~\eqref{eq:set_MPC} is feasible at $t=0$.  
Then Problem~\eqref{eq:set_MPC} is feasible for all $t\in\mathbb{N}$ and the constraints~\eqref{eq:constraint_condition}  are satisfied for the resulting closed loop system. 
\end{theorem}
\begin{proof}
The proof is similar to standard results in robust MPC, extended to the adaptive setting using a suitable monotonicity property in Assumption~\ref{ass:set_propagation}, compare e.g. \cite{Soloperto2019Collision} for similar arguments. 
In particular, consider any set $\mathbb{X}_{0|t+1}$ satisfying $\mathbb{X}^*_{1|t}\supseteq \mathbb{X}_{0|t+1} \ni x_{t+1}$, for example $\mathbb{X}_{0|t+1}=\mathbb{X}^*_{1|t}$, which satisfies $x_{t+1}\in\mathbb{X}_{0|t+1}$ due to the overapproximation property~\eqref{eq:set_inclusion}. 
Then  the candidate solution $u_{k|t+1}=u^*_{k+1|t}$, $u_{N|t+1}=u^*_{N|t}$, is a feasible solution with\footnote{%
In approaches directly utilizing sets, e.g.~\cite{lorenzen2019robust,Lu2019RAMPC,bayer2013discrete}, the candidate solution $\mathbb{X}_{k|t+1}=\mathbb{X}^*_{k+1|t}$ is standard. 
However, most tube-based approaches, e.g.~\cite{Koehler2019Adaptive,gonccalves2016robust,Robust_TAC_19,marruedo2002input,pin2009robust,limon2005robust}, including the proposed approach, consider a  candidate solution that does not necessarily satisfy $\mathbb{X}_{k|t+1}=\mathbb{X}^*_{k+1|t}$, but $\mathbb{X}_{k|t+1}\subseteq\mathbb{X}^*_{k+1|t}$. 
} $\mathbb{X}_{k|t+1}=\mathbb{X}^*_{k+1|t}$, $\mathbb{X}_{N|t+1}=\mathbb{X}_{N-1|t+1}=\mathbb{X}^*_{N|t}$, 
 due to the monotonicity of the operator $\Phi$~\eqref{eq:set_monotonicity}, the non-expansiveness of the parameter set $\Theta_t$ (Ass.~\ref{ass:param_set}) and the terminal constraint~\eqref{eq:set_MPC_term}. 
\end{proof}
The presented formulation in~\eqref{eq:set_MPC} and the corresponding theoretical properties in Theorem~\ref{thm:set_robust} are quite intuitive. 
However, without any tractable formulation for the sets $\mathbb{X}$, $\Theta$ and the propagation $\Phi$ this formulation cannot be used in practice. 
In this sense, the general formulation~\eqref{eq:set_MPC} and Theorem~\ref{thm:set_robust} are similar to the $\min$--$\max$ formulation in~\cite[Chap.~10.4]{guay2015robust}, as they provide general theoretical results but are not directly amenable to practical implementation. 

\subsubsection*{Existing tube formulations}
In the following, we briefly elaborate on different parametrizations for $\mathbb{X}$ and $\Theta$ which have been considered in the robust and robust adaptive MPC literature.

For linear systems, typically  polytopic sets $\Theta$ are considered and the tube is parametrized by a polytope $\mathbb{X}_t=\{x|~H x\leq \alpha_t\}$ with $H$ fixed offline. 
 The  inclusion~\eqref{eq:set_MPC_inclusion} can then be implemented using linear inequality constraints, compare e.g.~\cite{lorenzen2019robust,Lu2019RAMPC,Koehler2019Adaptive} and  \cite[Chap.~5]{kouvaritakis2016model}. 
In particular, the approach in~\cite{lorenzen2019robust} considers general polytopic parameter sets $\Theta$ and a homothetic tube approach, which directly formulates~\eqref{eq:set_MPC_inclusion} as equivalent linear inequality constraints, using additional dual variables $\Lambda$. 
In~\cite{Lu2019RAMPC},  zonotope parameter sets $\Theta_t=\{\Pi_\theta \theta\leq \pi_{t}\}$ with $\Pi_{\theta}$ fixed, are considered and the set inclusion~\eqref{eq:set_MPC_inclusion}  is implemented using $\alpha$ as optimization variables. 
In~\cite{Koehler2019Adaptive}, only scalars $s,\eta$ are used in parametrizing a hypercube parameter set $ \Theta_t=\overline{\theta}_t\oplus\eta_t \mathbb{B}_\infty$ and the polytopic tube $\mathbb{X}_t=\overline{x}_t\oplus s_t\cdot \mathcal{P}$. 
A more detailed  comparison regarding computational complexity and conservatism of these linear RAMPC approaches can be found in~\cite{Koehler2019Adaptive}.

In the papers~\cite{adetola2011robust,guay2015robust,gonccalves2016robust} by Martin Guay and coauthors, the parameter set $\Theta$ and the tube $\mathbb{X}$ are given by a scaled ball, i.e.,  $\Theta_t=\overline{\theta}_t\oplus\eta_t\mathbb{B}_2$ and $\mathbb{X}_t=\overline{x}_t\oplus s_t\mathbb{B}_2$. 
This simple scalar parametrization is crucial  in providing a tractable formulation, which allows the implementation of the set inclusion~\eqref{eq:set_MPC_inclusion} using scalar nonlinear dynamics for $s$. 
However, this can also yield very conservative bounds on the tube size $s$ along the prediction horizon, compare the numerical example in Section~\ref{sec:num}. 

In~\cite{lopez2018adaptive}, a hyper box $\Theta$ and a box shaped tube $\mathbb{X}$ are considered resulting from a boundary layer controller, while the boundary layer thickness $s$ is predicted using nonlinear dynamics for $s$, compare~\cite{DynamicTube_Lopez_19}.

Regarding general robust MPC schemes for nonlinear systems (without parameter adaptation and often without parametric uncertainty): 
In~\cite{limon2005robust} an interval arithmetic approach is considered, which improves Lipschitz based approaches~\cite{marruedo2002input,pin2009robust} by using a more flexible hyperbox tube $\mathbb{X}$. 
However, similar to the Lipschitz based approach, unless the considered tube parametrization $\mathbb{X}$ contains a robust positive invariant (RPI) set, the tube is growing unbounded along the prediction horizon and thus only short horizons and/or small uncertainty can be considered. 
In~\cite{singh2017robust,bayer2013discrete} additive disturbances are considered and a fixed RPI set $\mathbb{X}$  is  computed offline  as an incremental Lyapunov function or using control contraction metrics. 

In~\cite{villanueva2017robust}, the tube is parametrized with online optimized matrices $P_t\in\mathbb{R}^{n\times n}$, i.e., $\mathbb{X}_t=\{x|~\|\overline{x}_t-x\|_{P_t}^2\leq 1\}$, where \eqref{eq:set_MPC_inclusion} is ensured using $\min$--$\max$ differential inequalities.

In~\cite{Robust_TAC_19}, a scalar $s$ is used to parametrize the tube $\mathbb{X}$ with a given incremental Lyapunov function $V_{\delta}$, i.e., $\mathbb{X}_t=\{x|~V_{\delta}(x,\overline{x}_t)\leq s_t\}$, while the tube propagation~\eqref{eq:set_MPC_inclusion} is formulated as nonlinear dynamics for $s$, similar to~\cite{DynamicTube_Lopez_19,pin2009robust}.   
Thus, this approach shares the simple scalar characterization of the tube used in~\cite{adetola2011robust,guay2015robust,gonccalves2016robust}. %, but does not 
In the next section, we concretize the rather abstract conditions and assumptions for the specific parametrization of  $\mathbb{X}$ and $\Phi$ based on the robust MPC approach in~\cite{Robust_TAC_19}.

%% file: Theory.tex
%!TEX root = ./Adaptive_Nonlin.tex
%%%%%%%%%%%%%%%%%%%%%%%%%%%%%%%%%%%%%%%%%%%%%%%%%%%%%%%%%%%%%%%%%%%%%%%%%%%%%%%
\section{Proposed framework - theoretical analysis} 
\label{sec:main}
The results in this section are the main contribution of this paper. 
In the following, we derive the proposed framework for nonlinear RAMPC using the nonlinear robust MPC framework in~\cite{Robust_TAC_19}.
In Section~\ref{sec:main_2} the general conditions and assumptions are introduced.
The RAMPC optimization problem is presented in Section~\ref{sec:main_1}. 
Constraint satisfaction and robust recursive feasibility are established in Theorem~\ref{thm:main} in Section~\ref{sec:main_3}. 
Theorem~\ref{thm:stability} in Section~\ref{sec:main_4} shows finite gain $\mathcal{L}_2$ stability using a least mean square (LMS) point estimate and a suitable stage cost $\ell:\mathcal{Z}\rightarrow\mathbb{R}$.
Section~\ref{sec:design_1} provides explicit design procedures and the overall algorithm.
In Section~\ref{sec:special_guay} we demonstrate that state of the art approaches~\cite{adetola2011robust,guay2015robust,gonccalves2016robust} are contained as a special case of the proposed formulation. 
Section~\ref{sec:main_6} discusses some extensions and open issues.
%define conditions
\input{Theory_2}

%algorithm
\input{Theory_1}

%theorem
\input{Theory_3}

%stability
\input{Theory_4}
 %design - algorithms
\input{Approaches_1}

%\newpage
\input{Approaches_2} 
%extensions
\input{Theory_6}

%% file: Theory_2.tex
%!TEX root = ./Adaptive_Nonlin.tex
%%%%%%%%%%%%%%%%%%%%%%%%%%%%%%%%%%%%%%%%%%%%%%%%%%%%%%%%%%%%%%%%%%%%%%%%%%%%%%%
\subsection{Assumptions}
\label{sec:main_2}
In the following, we introduce assumptions regarding the nonlinear system $f_w$, including model structure (Ass.~\ref{ass:model}), nominal parameter updates (Ass.~\ref{ass:nominal}) , stabilizability (Ass.~\ref{ass:increm}), designed functions $\tilde{w}$ (Ass.~\ref{ass:w_tilde}) and terminal ingredients (Ass.~\ref{ass:term}). 

The following standing assumption characterizes the considered class of nonlinear systems.
\begin{assumption}
\label{ass:model}
There exist (locally) Lipschitz continuous functions $g_{i}:\mathcal{Z}\rightarrow\mathbb{R}^n$, $i=1,\dots,p$ and a matrix $E\in\mathbb{R}^{n\times q}$, such that the nonlinear system~\eqref{eq:dyn} is given by
\begin{align}
\label{eq:model_affine}
f_w(x,u,d,\theta)=&f(x,u)+G(x,u)\theta+Ed, \\
G(x,u):=&[g_1(x,u),\dots , g_p(x,u)].\nonumber
\end{align}
There exists a known set $\mathbb{D}\subset\mathbb{R}^q$ with $0\in\mathbb{D}$, such that the additive disturbance satisfy $d_t\in\mathbb{D}$ for all $t\geq 0$. \\
 The functions $h_j$ in~\eqref{eq:constraint_definition} are (locally) Lipschitz continuous. 
\end{assumption}	
The main restriction in the posed conditions is that the parameters $\theta$ enter affinely and the disturbances are only additive. 
Except for the paper~\cite{adetola2014adaptive} which explicitly handles the intricate case of nonlinearly parametrized systems, most existing RAMPC schemes for linear~\cite{tanaskovic2014adaptive,lorenzen2019robust,Lu2019RAMPC,Koehler2019Adaptive,bujarbaruah2019adaptive} and nonlinear systems~\cite{lopez2018adaptive,adetola2011robust,guay2015robust,gonccalves2016robust,dahliwal2014set} also consider a linear parametrization in $\theta$ and additive disturbances $d_t$.
In Section~\ref{sec:main_6} we discuss how to relax the Lipschitz continuity  of $h_j$ and extend the proposed approach to time varying parameters $\theta_t^*$. 

In the following, we consider a nominal prediction model $f_{\overline{\theta}_t}(x,u):=f(x,u)+G(x,u)\overline{\theta}_t$ with online determined parameters $\overline{\theta}_t$ satisfying the following assumption. 
\begin{assumption}
\label{ass:nominal}
At each time $t$, we compute a point estimate $\overline{\theta}_t$ with a corresponding uncertainty set $\widetilde{\Theta}_t$ that satisfy
\begin{align}
\label{eq:nominal_param}
\theta^*\in\overline{\theta}_{t+1}\oplus\widetilde{\Theta}_{t+1}\subseteq \overline{\theta}_t\oplus\widetilde{\Theta}_t,~\forall t\geq 0.
\end{align}
with some initial known prior parameter set $\overline{\theta}_0\oplus\widetilde{\Theta}_0$. 
\end{assumption}
In essence, this assumption is equivalent to Assumption~\ref{ass:param_set}, formulated in terms of a nominal point $\overline{\theta}$ and an uncertain set $\widetilde{\Theta}$. 
Correspondingly, the change in parameters satisfies $\Delta\overline{\theta}_t:=\overline{\theta}_{t+1}-\overline{\theta}_t\in\Delta\widetilde{\Theta}_t:=\widetilde{\Theta}_t\ominus\widetilde{\Theta}_{t+1}$. 
Set membership updates for $\widetilde{\Theta}_t,\overline{\theta}_t$ satisfying Assumption~\ref{ass:nominal} will be introduced in Sec.~\ref{sec:design_1}, Alg.~\ref{alg:HC}. 
The prediction mismatch satisfies 
\begin{align}
\label{eq:model_mismatch}
&x_{t+1}-f_{\overline{\theta}_t}(x_t,u_t)\in\mathbb{W}_{\widetilde{\Theta}_t,\mathbb{D}}(x_t,u_t), \\
&\mathbb{W}_{\widetilde{\Theta},\mathbb{D}}(x,u):=\{d_w\in\mathbb{R}^n|~d_w=d+G(x,u)\tilde{\theta}, ~d\in\mathbb{D},~\tilde{\theta}\in\widetilde{\Theta}\}.\nonumber
\end{align}

In order to design a suitable tube $\mathbb{X}$,  we assume that the system is locally incrementally
stabilizable, similar to~\cite[Ass.~2]{Robust_TAC_19}.
\begin{assumption}
\label{ass:increm}
There exists a continuous incremental Lyapunov function $V_{\delta}:\mathbb{R}^n\times\mathbb{R}^n\rightarrow\mathbb{R}_{\geq 0}$ satisfying 
\begin{subequations}
\label{eq:increm}
\begin{align}
\label{eq:increm_a}
c_{\delta,l}\|x-z\|\leq V_{\delta}(x,z)\leq& c_{\delta,u}\|x-z\|,
\end{align}
for all $x,z\in\mathbb{R}^n$ with constants $c_{\delta,l}$, $c_{\delta,u}>0$.
Furthermore, there exist a control law $\kappa:\mathcal{Z}_x\times\mathcal{Z}\rightarrow\mathbb{R}^m$, and constants $\delta_{loc}$, $\kappa_{\max}>0$, such that the following properties hold for all $(z,v)\in\mathcal{Z}$, $(x,\kappa(x,z,v))\in\mathcal{Z}$, $V_{\delta}(x,z)\leq \delta_{loc}$, and all $\overline{\theta}\in\overline{\theta}_0\oplus\widetilde{\Theta}_0$:
\begin{align}
\label{eq:increm_b}
\|\kappa(x,z,v)-v\|\leq &\kappa_{\max}V_{\delta}(x,z),\\
\label{eq:increm_c}
V_{\delta}(x^+,z^+)\leq &\rho_{\overline{\theta}}V_{\delta}(x,z),
\end{align}
with $x^+=f_{\overline{\theta}}(x,\kappa(x,z,v))$, $z^+=f_{\overline{\theta}}(z,v)$ and some $\rho_{\overline{\theta}}>0$. 
Furthermore, the following norm-like condition holds for any $x_1,x_2,\Delta x\in\mathbb{R}^n$: 
\begin{align}
\label{eq:increm_d}
&V_{\delta}(x_1+\Delta x,x_2)\leq V_{\delta}(x_1,x_2)+V_{\delta}(x_2+\Delta x,x_2).
\end{align} 
In addition, there exists a constant $L_{\delta}\geq 0$, such that the following continuity bound holds for any $x_1,x_2,\Delta x\in\mathbb{R}^n$:
\begin{align}
\label{eq:increm_e}
&V_{\delta}(x_1,x_2+\Delta x)
\leq (1+L_{\delta}\|\Delta x\|)V_{\delta}(x_1-\Delta x,x_2).
\end{align} 
\end{subequations}
\end{assumption}	
In case $\rho_{\overline{\theta}}<1$, conditions~\eqref{eq:increm_a}--\eqref{eq:increm_c} imply that $V_{\delta}$ is an incremental exponential Lyapunov function with some Lipschitz continuous feedback $\kappa$.  
A detailed discussion how existing tube parametrizations are related to Assumption~\ref{ass:increm} can be found in~\cite[Remark~1]{Robust_TAC_19}. 
The norm-like inequality~\eqref{eq:increm_d} and Lipschitz like condition~\eqref{eq:increm_e} are, for example, satisfied by polytopes $V_{\delta}(x,z)=\max_i P_i(x-z)$, ellipsoids $V_{\delta}(x,z)=\|x-z\|_P$ and functions of the form\footnote{%
A proof of inequality~\eqref{eq:increm_e} for $P(z)$ Lipschitz continuous can be found in~\cite[Prop.~15]{JK_periodic_automatica}. 
Examples of such incremental Lyapunov functions can be found in~\cite{Soloperto2019Collision,JK_periodic_automatica,JK_QINF} using a quasi-LPV parametrization, compare~\cite{JK_QINF} for a corresponding LMI design procedure. 
} $V_{\delta}(x,z)=\|x-z\|_{P(z)}$.

In the following, we denote the set 
$\Psi:=\{(x,z,v)\in\mathbb{R}^n\times\mathcal{Z}|~(x,\kappa(x,z,v))\in\mathcal{Z},~V_{\delta}(x,z)\leq \delta_{loc}\}$. 
For each constraint~\eqref{eq:constraint_definition}, we compute constants $c_j\geq 0$, $j=1,\dots,r$ satisfying
\begin{align}
\label{eq:c_j}
h_j(x,\kappa(x,z,v))-h_j(z,v)\leq c_j V_{\delta}(x,z),
\end{align}
for all $(x,z,v)\in\Psi$, which will later be used in the design.
Existence of finite constants $c_j$ satisfying~\eqref{eq:c_j} follows from $h_j$ Lipschitz continuous and the bound on $V_{\delta},\kappa$ in~\eqref{eq:increm_a}--\eqref{eq:increm_b}. 

The smallest contraction rate $\rho_{\theta}$ for a given value of $\theta\in\Theta_0$ satisfying~\eqref{eq:increm_c} is given by
\begin{align}
\label{eq:rho_theta}
\rho_{\theta}:=\max_{(x,z,v)\in\Psi}\dfrac{V_{\delta}(f_{\theta}(x,\kappa(x,z,v)),f_{\theta}(z,v))}{{V_{\delta}(x,z)}}.
\end{align}
The following proposition provides a bound on the change of $\rho_{\theta}$ under changing parameters. 
\begin{proposition}
%Prop.~4.4
\label{prop:properties_2}
Let Assumptions~\ref{ass:model} and \ref{ass:increm} hold.  
For any $\theta$, $\Delta\Theta$ there exists a Lipschitz constant $L_{\rho,\theta,\Delta\Theta}\geq 0$ according to~\eqref{eq:L_rho_def}, such that for any  ${\theta}^+\in{\theta}\oplus\Delta{\Theta}\subseteq\Theta_0$,  the contraction rate $\rho_{\theta}$ satisfies 
\begin{align}
\label{eq:rho_Lipschitz}
\rho_{{\theta}^+} \leq \rho_{{\theta}}+L_{\rho,\theta,\Delta{\Theta}}.
\end{align}
\end{proposition}
\begin{proof}
First, note that for any $x_1,x_2,\Delta x_1,\Delta x_2\in\mathbb{R}^n$  the following continuity condition holds
\begin{align}
\label{eq:increm_f}
&V_{\delta}(x_1+\Delta x_1,x_2+\Delta x_2)\\
 \stackrel{\eqref{eq:increm_e}}{\leq}& (1+L_\delta\|\Delta x_2\|)V_{\delta}(x_1+\Delta x_1-\Delta x_2,x_2)\nonumber\\
\stackrel{\eqref{eq:increm_d}}{\leq}&(1+L_{\delta}\|\Delta x_2\|)(V_{\delta}(x_1,x_2)+V_{\delta}(x_2+\Delta x_1-\Delta x_2,x_2))\nonumber\\
\stackrel{\eqref{eq:increm_a}}{\leq}&(1+L_{\delta}\|\Delta x_2\|)(V_{\delta}(x_1,x_2)+c_{\delta,u}\|\Delta x_1-\Delta x_2\|).\nonumber
\end{align}
Denote $\Delta \theta=\theta^+-\theta$. 
For any $(x,z,v)\in\Psi$, we have
\begin{align*}
&V_{\delta}((f_{\theta^+}(x,\kappa(x,z,v)),f_{\theta^+}(z,v)))\\
\stackrel{\eqref{eq:increm_f}}{\leq}&(1+L_{\delta}\|G(z,v)\Delta \theta\|)(V_{\delta}(f_{\theta}(x,\kappa(x,z,v)),f_{\theta}(z,v))\\
&+c_{\delta,u}\|(G(x,\kappa(x,z,v))-G(z,v))\Delta\theta\|)\\
\stackrel{\eqref{eq:rho_theta}}{\leq} &(1+L_{\delta}\overline{G}\|\Delta \theta\|)(\rho_\theta+c_{\delta,u}L_{G,\kappa}\|\Delta\theta\|)V_{\delta}(x,z),
\end{align*}
where $\overline{G}=\max_{(z,v)\in\mathcal{Z}}\|G(z,v)\|$ and $L_{G,\kappa}$ is a suitable Lipschitz constant, given $G,\kappa$ Lipschitz and the lower bound in~\eqref{eq:increm_a}.
Thus, 
\begin{align}
\label{eq:L_rho_def}
L_{\rho,\Delta\Theta}:=(L_{\delta}\overline{G}\rho_\theta+c_{\delta,u}L_{G,\kappa})\epsilon_{\Delta\Theta}+L_{\delta}\overline{G}c_{\delta,u}L_{G,\kappa}\epsilon_{\Delta\Theta}^2,
\end{align}
with $\epsilon_{\Delta\Theta}=\max_{\theta\in\Delta\Theta}\|\theta\|$  satisfies~\eqref{eq:rho_Lipschitz}. 
\end{proof}
In Sec.~\ref{sec:design_1}, Prop.~\ref{prop:properties_2_quadratic} for the special case of $V_{\delta}(x,z)=\|x-z\|_P$, we will derive a simpler expression for $L_{\rho,\theta,\Delta\Theta}$. 

In order to facilitate an efficient evaluation of the uncertainty (possible model mismatch) at some point $(z,v)\in\mathcal{Z}$ or in a neighbourhood thereof, we design a function $\tilde{w}$ offline, satisfying the following conditions.
\begin{assumption}
\label{ass:w_tilde}
Consider the functions $V_{\delta},\kappa$ from Assumption~\ref{ass:increm} and the Lipschitz constant $L_{\rho,\overline{\theta},\Delta\widetilde{\Theta}}$ from Prop.~\ref{prop:properties_2}. 
For any sets $\widetilde{\Theta}^+$, $\Delta\widetilde{\Theta}$, $\widetilde{\Theta}$ and parameters $\overline{\theta}$, such that $\widetilde{\Theta}^+\oplus\Delta\widetilde{\Theta}\subseteq\widetilde{\Theta}$,  $\overline{\theta}\oplus\widetilde{\Theta}\subseteq\Theta_0$, 
there exist a \textit{scalar} disturbance bound $\tilde{w}_{\widetilde{\Theta},\mathbb{D}}:\mathcal{Z}\times\mathbb{R}_{\geq 0}\rightarrow\mathbb{R}_{\geq 0}$ and a constant $L_{\widetilde{\Theta}}\geq 0$, such that the following properties hold for all $(x,z,v)\in\Psi$, any model mismatch $d_w\in\mathbb{W}_{\widetilde{\Theta},\mathbb{D}}(z,v)$, any state $\tilde{z}\in\mathcal{Z}_x$: 
\begin{subequations}
\begin{align}
\label{eq:w_tilde_a}
&V_{\delta}(\tilde{z}+d_w,\tilde{z})\leq \tilde{w}_{\widetilde{\Theta},\mathbb{D}}(z,v),\\
\label{eq:w_tilde_b}
& \tilde{w}_{\widetilde{\Theta},\mathbb{D}}(x,\kappa(x,z,v))- \tilde{w}_{\widetilde{\Theta},\mathbb{D}}(z,v)\leq L_{\widetilde{\Theta}}V_{\delta}(x,z),\\
\label{eq:w_tilde_c}
&\tilde{w}_{\widetilde{\Theta},\mathbb{D}}(z,v)\geq \tilde{w}_{\widetilde{\Theta}^+,\mathbb{D}}(z,v)+\tilde{w}_{\Delta\widetilde{\Theta},\{0\}}(z,v),\\
\label{eq:w_tilde_d}
&L_{\widetilde{\Theta}}\geq L_{\widetilde{\Theta}^+}+L_{\Delta\widetilde{\Theta}},\\
\label{eq:w_tilde_e}
&L_{\rho,\overline{\theta},\Delta\widetilde{\Theta}}\leq L_{\Delta\widetilde{\Theta}}.
\end{align}
\end{subequations}
\end{assumption}
These conditions are a generalization of~\cite[Ass.~5, Prop.~2]{Robust_TAC_19} to the adaptive setting.
Condition~\eqref{eq:w_tilde_a} provides an upper bound on the model mismatch and will later be used to design a function $\Phi$ that ensures satisfaction of~\eqref{eq:set_inclusion} for $\mathbb{X}=\{\overline{x}\}$. 
Condition~\eqref{eq:w_tilde_b} provides a Lipschitz bound $L_{\widetilde{\Theta}}$ on $\tilde{w}$. 
Conditions~\eqref{eq:w_tilde_c} and \eqref{eq:w_tilde_d} provide a monotonicity property w.r.t. the parametric uncertainty. 
 Condition~\eqref{eq:w_tilde_e} can always be ensured by choosing $L_{\widetilde{\Theta}}$ large enough. 
 Conditions~~\eqref{eq:w_tilde_d} and \eqref{eq:w_tilde_e} imply that the possible increase in $\rho_\theta$ due to parameter updates is smaller than the decrease in   $L_{\widetilde{\Theta}}$ (and thus $\tilde{w}$). 
Corresponding designs will be introduced in Sec.~\ref{sec:design_1}, Prop.~\ref{prop:design_w_1}, \ref{prop:design_w_2}.

The following proposition defines a function $\tilde{w}_{\delta}$, which shares similar monotonicity and overapproximation properties, but holds for all points in a neighbourhood of size $s$ around a given point $(z,v)\in\mathcal{Z}$.
\begin{proposition}
\label{prop:w_tilde}
Let Assumptions~\ref{ass:increm} and \ref{ass:w_tilde} hold. 
Define 
\begin{subequations}
\label{eq:w_delta_tilde}
\begin{align}
\label{eq:def_w_delta_tilde}
\tilde{w}_{\delta,\widetilde{\Theta},\mathbb{D}}(z,v,s):=\tilde{w}_{\widetilde{\Theta},\mathbb{D}}(z,v)+L_{\widetilde{\Theta}}s.
\end{align}
For any $(x,z,v)\in\Psi$, $\Delta s\geq 0$ with $V_{\delta}(x,z)\leq \Delta s$, we have
\begin{align}
\label{eq:w_tilde_g}
\tilde{w}_{\delta,\widetilde{\Theta},\mathbb{D}}(x,\kappa(x,z,v),s)\leq \tilde{w}_{\delta,\widetilde{\Theta},\mathbb{D}}(z,v,s+\Delta s).
\end{align}
Furthermore, for any $\widetilde{\Theta}^+\oplus\Delta\widetilde{\Theta}\subseteq\widetilde{\Theta}\subseteq\Theta_0$, we have
\begin{align}
\label{eq:w_tilde_c_delta}
&\tilde{w}_{\delta,\widetilde{\Theta},\mathbb{D}}(z,v,s)\geq \tilde{w}_{\delta,\widetilde{\Theta}^+,\mathbb{D}}(z,v,s)+\tilde{w}_{\delta,\Delta\widetilde{\Theta},\{0\}}(z,v,s).
\end{align}
\end{subequations}
\end{proposition}
\begin{proof}
Condition~\eqref{eq:w_tilde_g} follows from the definition~\eqref{eq:def_w_delta_tilde} and condition~\eqref{eq:w_tilde_b}:
\begin{align*}
\tilde{w}_{\delta,\tilde{\Theta},\mathbb{D}}(x,\kappa(x,z,v),s)\stackrel{\eqref{eq:def_w_delta_tilde}}{=}\tilde{w}_{\tilde{\Theta},\mathbb{D}}(x,\kappa(x,z,v))+L_{\tilde{\Theta}}s\\
\stackrel{\eqref{eq:w_tilde_b}}{\leq} \tilde{w}_{\tilde{\Theta},\mathbb{D}}(z,v)+L_{\tilde{\Theta}}(s+\Delta s)\stackrel{\eqref{eq:def_w_delta_tilde}}{=}\tilde{w}_{\delta,\tilde{\Theta},\mathbb{D}}(z,v,s+\Delta s).
\end{align*}
Condition~\eqref{eq:w_tilde_c_delta} follows directly from~\eqref{eq:w_tilde_c} and \eqref{eq:w_tilde_d}. 
\end{proof}
In the following, we denote the minimal uncertainty (due to additive disturbances) by $\overline{d}:=\min_{(x,u)\in\mathcal{Z}}\tilde{w}_{\{0\},\mathbb{D}}(x,u)$.

The following assumption captures the desired properties of the terminal set.
\begin{assumption}
\label{ass:term}
Consider $V_{\delta},\tilde{w}_{\delta,\widetilde{\Theta},\mathbb{D}}$ from Assumptions~\ref{ass:increm}, \ref{ass:w_tilde} and Prop.~\ref{prop:w_tilde}. 
There exist a control law $k_f:\mathbb{R}^n\rightarrow\mathbb{R}^m$, a terminal cost $V_{f}:\mathbb{R}^n\rightarrow\mathbb{R}_{\geq 0}$, a function $\alpha_v\in\mathcal{K}_{\infty}$,  scalars $\overline{w}_{\widetilde{\Theta}}\geq 0$, $\overline{s} \in(0,\delta_{loc}]$, and a terminal region $\mathcal{X}_{f,\overline{\theta},\widetilde{\Theta}}\subseteq\mathbb{R}^{n+1}$ such that for all 
\begin{enumerate}[label=\alph*)]
\item  $(x,s)\in\mathcal{X}_{f,\overline{\theta},\widetilde{\Theta}}$,
\label{cond_term_a}
\item $\overline{\theta}^+\oplus\widetilde{\Theta}^+\subseteq\overline{\theta}\oplus\widetilde{\Theta}\subseteq\overline{\theta}_0\oplus\widetilde{\Theta}_0$, $\Delta\widetilde{\Theta}:=\widetilde{\Theta}\ominus\widetilde{\Theta}^+$,
\label{cond_term_b}
\item  $\tilde{s}:$ $(\rho_{\overline{\theta}}+L_{\Delta\widetilde{\Theta}})^N\overline{d}\leq \tilde{s}\\
\leq(\rho_{\overline{\theta}}+L_{\Delta\widetilde{\Theta}})^N\overline{w}_{\widetilde{\Theta}}+(\overline{w}_{\Delta\widetilde{\Theta}}-\overline{d})\sum_{k=0}^{N-1}(\rho_{\overline{\theta}}+L_{\Delta\widetilde{\Theta}})^k$,
\label{cond_term_c}
\item $s^+\in[0,\rho_{\overline{\theta}}s+\tilde{w}_{\delta,\widetilde{\Theta},\mathbb{D}}(x,k_f(x),s)-\tilde{s}]$,
\label{cond_term_d}
\item$x^+\in\mathbb{R}^n$: $V_{\delta}(x^+,f_{\overline{\theta}}(x,k_f(x)))\leq \tilde{s}$
\label{cond_term_e}
\end{enumerate}
the following properties hold 
\begin{subequations}
\begin{align}
\label{eq:term_1}
(x^+,s^+)\in&\mathcal{X}_{f,\overline{\theta}^+,\widetilde{\Theta}^+},\\
\label{eq:term_2}
h_j(x,k_f(x))+c_j s\leq& 0,~j=1,\dots ,r,\\
\label{eq:term_3}
\tilde{w}_{\delta,\widetilde{\Theta},\mathbb{D}}(x,k_f(x),s)\leq &\overline{w}_{\widetilde{\Theta}},\\
\label{eq:term_4}
s\leq &\overline{s},\\
\label{eq:term_5}
V_{f}(x^+)-V_{f}(x)\leq &-\ell(x,k_f(x))+\alpha_v(\tilde{s}).
\end{align}
Furthermore, for any $\widetilde{\Theta}^+\subseteq\widetilde{\Theta}\subseteq\widetilde{\Theta}_0$ and any $(x,u,s)\in\mathcal{Z}\times\mathbb{R}_{\geq 0}$ the following implication holds:
\begin{align}
\label{eq:term_6}
\tilde{w}_{\delta,\widetilde{\Theta},\mathbb{D}}(x,u,s)\leq \overline{w}_{\widetilde{\Theta}} ~\Rightarrow~ \tilde{w}_{\delta,\widetilde{\Theta}^+,\mathbb{D}}(x,u,s)\leq \overline{w}_{\widetilde{\Theta}^+}.
\end{align}
\end{subequations}
\end{assumption}
The conditions~\ref{cond_term_c}-\ref{cond_term_e} on $\tilde{s},s^+,x^+$ directly follow from the candidate solution used in Theorem~\ref{thm:main} below.
Property~\eqref{eq:term_1} ensures recursive feasibility of the terminal constraint. 
Properties~\eqref{eq:term_2}--\eqref{eq:term_4} ensure that the tightened constraints are satisfied in the terminal region. %, \eqref{eq:term_3}, 
Property~\eqref{eq:term_5} is not needed for recursive feasibility and constraint satisfaction but will be used to provide suitable stability guarantees. 
Property~\eqref{eq:term_6} ensures that the bound $\overline{w}_{\widetilde{\Theta}}$ can be reduced if the uncertainty set $\widetilde{\Theta}$ shrinks, but only such that previously feasible trajectories $(x,u,s)$ remain feasible.
A constructive design procedure satisfying these conditions will be introduced in Sec.~\ref{sec:design_1}, Prop.~\ref{prop:terminal_design}. 
We would like to point out that the presented terminal conditions are more intricate than the simple robust control invariance conditions considered in the works by Martin Guay and coauthors~\cite{adetola2011robust,guay2015robust,gonccalves2016robust}.
However, to the best knowledge of the authors, these simpler terminal conditions are not sufficient to prove recursive feasibility for tube-based RAMPC schemes and only apply to conceptual $\min$--$\max$ RAMPC approaches.

%% file: Theory_1.tex
%!TEX root = ./Adaptive_Nonlin.tex
%%%%%%%%%%%%%%%%%%%%%%%%%%%%%%%%%%%%%%%%%%%%%%%%%%%%%%%%%%%%%%%%%%%%%%%%%%%%%%%
\subsection{Proposed RAMPC formulation}
\label{sec:main_1}
In the following, we specify the optimization problem for the proposed RAMPC approach, given the nominal prediction model $f_{\theta}$, parameters and set $\widetilde{\Theta}_t,\overline{\theta}_t$ (Ass.~\ref{ass:nominal}), the controller $\kappa$ (Ass.~\ref{ass:increm}) and the terminal ingredients (Ass.~\ref{ass:term}).
At time $t$, given the measured state $x_t$, set $\widetilde{\Theta}_t$, nominal parameters $\overline{\theta}_t$ and a later specified point estimate $\hat{\theta}_t$, the optimization problem is given by:
\begin{subequations}
\label{eq:RAMPC}
\begin{align}
\label{eq:RAMPC_cost}
&\min_{\overline{u}_{\cdot|t},w_{\cdot|t}} \sum_{k=0}^{N-1}\ell(\hat{x}_{k|t},\hat{u}_{k|t})+V_{f}(\hat{x}_{N|t})\\
\label{eq:RAMPC_init}
\text{s.t. } &\hat{x}_{0|t}=\overline{x}_{0|t}=x_t,~s_{0|t}=0,\\
\label{eq:RAMPC_dyn}
&\overline{x}_{k+1|t}=f_{\overline{\theta}_t}(\overline{x}_{k|t},\overline{u}_{k|t}),~\hat{x}_{k+1|t}=f_{\hat{\theta}_t}(\hat{x}_{k|t},\hat{u}_{k|t}),\\
\label{eq:RAMPC_dyn_s}
&s_{k+1|t}=\rho_{\overline{\theta}_t} s_{k|t}+w_{k|t},\\
\label{eq:RAMPC_w}
&w_{k|t}\geq \tilde{w}_{\delta,\widetilde{\Theta}_t,\mathbb{D}}(\overline{x}_{k|t},\overline{u}_{k|t},s_{k|t}),\\
\label{eq:RAMPC_con}
&h_j(\overline{x}_{k|t},\overline{u}_{k|t})+c_js_{k|t}\leq 0,\\
\label{eq:RAMPC_u_hat}
&\hat{u}_{k|t}=\kappa(\hat{x}_{k|t},\overline{x}_{k|t},\overline{u}_{k|t}),\\
\label{eq:RAMPC_con_sw}
&s_{k|t}\leq \overline{s},\quad w_{k|t}\leq \overline{w}_{\widetilde{\Theta}_t},\\
\label{eq:RAMPC_term}
&(\overline{x}_{N|t},s_{N|t})\in\mathcal{X}_{f,\overline{\theta}_t,\widetilde{\Theta}_t},\\
&k=0,\dots,N-1,\quad j=1,\dots,r.\nonumber
\end{align}
\end{subequations}
The minimizers are denoted by $\overline{u}^*_{\cdot|t}$, $w^*_{\cdot|t}$ with $\overline{x}^*_{\cdot|t},\hat{x}^*_{\cdot|t},\hat{u}^*_{\cdot|t},s^*_{\cdot|t}$ according to~\eqref{eq:RAMPC_init}, \eqref{eq:RAMPC_dyn},  \eqref{eq:RAMPC_dyn_s}, \eqref{eq:RAMPC_u_hat}  and the corresponding value function $V_t$. 
In closed-loop operation the optimization problem~\eqref{eq:RAMPC} is solved at each time step $t\in\mathbb{N}$ and the input $u_t=\overline{u}^*_{0|t}$ is applied to the system yielding the following closed-loop system
\begin{align}
\label{eq:close}
x_{t+1}=f_{\theta^*}(x_t,\overline{u}^*_{0|t})+Ed_t\in\{\overline{x}^*_{1|t}\}\oplus\mathbb{W}_{\widetilde{\Theta}_t,\mathbb{D}}(x_t,u_t).
\end{align}
In the following, we explain the different elements in~\eqref{eq:RAMPC}. 
The trajectory $\overline{x},\overline{u}$~\eqref{eq:RAMPC_init}, \eqref{eq:RAMPC_dyn} corresponds to a nominal predicted trajectory, while robust constraint satisfaction is ensured by using tightened constraints~\eqref{eq:RAMPC_con} based on the predicted tube size $s$~\eqref{eq:RAMPC_dyn_s} and the uncertainty $w$~\eqref{eq:RAMPC_w}.
In particular, we can define the predicted tube $\mathbb{X}_{k|t}:=\{x|~V_{\delta}(x,\overline{x}^*_{k|t})\leq s_{k|t}^*\}$ with corresponding input $\mathbb{U}_{k|t}:=\kappa(\mathbb{X}_{k|t},\overline{x}_{k|t}^*,\overline{u}^*_{k|t})$. 
Thus, the constraints~\eqref{eq:RAMPC_con} with $c_j$ according to~\eqref{eq:c_j} directly guarantees $\mathbb{X}_{k|t}\times\mathbb{U}_{k|t}\subseteq\mathcal{Z}$, similar to~\eqref{eq:set_MPC_tightened_constraints}. 
Furthermore, the dynamics of the nominal system $\overline{x}$ and the tube $s$ correspond to the general tube propagation $\Phi$ used in~\eqref{eq:set_MPC_inclusion}. 
The monotonicity and overapproximation property (Ass.~\ref{ass:set_propagation}) are ensured by the posed conditions (Ass.~\ref{ass:increm}, \ref{ass:w_tilde}), which will be shown in Theorem~\ref{thm:main}. 
The terminal constraint~\eqref{eq:RAMPC_term} in combination with the condition in Assumption~\ref{ass:term} ensures that the condition~\eqref{eq:set_MPC_term} is satisfied. 
The constraints~\eqref{eq:RAMPC_con_sw} limit the tube size $s$ and uncertainty $w$, which may have practical benefits (avoiding regions with large uncertainty) and can be useful in some design for the terminal ingredients.
The trajectory $\hat{x},\hat{u}$ is used for improved stability properties based on a stage cost $\ell$ and a later specified LMS point estimate $\hat{\theta}_t\in\overline{\theta}_t\oplus\widetilde{\Theta}_t$ (c.f. Sec.~\ref{sec:main_4}), which corresponds to one specific trajectory in the predicted tube, i.e., $(\hat{x}_{k|t},\hat{u}_{k|t})\in\mathbb{X}_{k|t}\times\mathbb{U}_{k|t}$, compare Theorem~\ref{thm:stability}.   

If we compare this formulation in terms of computational complexity to a nominal MPC, we have additional decision variables $w_{\cdot|t}$ and additional inequality constraints~\eqref{eq:RAMPC_w}, \eqref{eq:RAMPC_con_sw}. 
In particular, the proposed formulation is equivalent to a nominal MPC scheme with an augmented state $(x,s)\in\mathbb{R}^{n+1}$, augmented input vector $(u,w)\in\mathbb{R}^{m+1}$ and additional nonlinear inequality constraints~\eqref{eq:RAMPC_w}. 
The fact that the parameter set $\widetilde{\Theta}_t$ and parameters $\overline{\theta}_t$ are updated online has no impact on the computational demand of solving~\eqref{eq:RAMPC}.

%% file: Theory_3.tex
%!TEX root = ./Adaptive_Nonlin.tex
%%%%%%%%%%%%%%%%%%%%%%%%%%%%%%%%%%%%%%%%%%%%%%%%%%%%%%%%%%%%%%%%%%%%%%%%%%%%%%%
\subsection{Main Theorem - robust recursive feasibility}
\label{sec:main_3}
\begin{figure}[tbp]
\begin{center}
\includegraphics[width=0.35\textwidth]{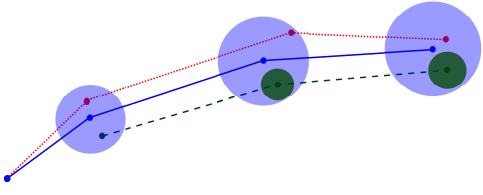}
\end{center}
\caption{Illustration - nested tubes property: Optimal trajectory $\overline{x}^*_{\cdot|t}$ (blue, solid), candidate trajectory $\overline{x}_{\cdot|t+1}$ (green, dashed), LMS trajectory $\hat{x}^*_{\cdot|t}$ (red, dotted), with corresponding tubes  $\mathbb{X}^*_{k|t}=\{z|V_{\delta}(z,\overline{x}^*_{k|t})\leq s^*_{k|t}\}$ (blue ellipses), $\mathbb{X}_{k|t+1}=\{\tilde{z}|~V_{\delta}(\tilde{z},\overline{x}_{k|t+1})\leq s_{k|t+1})\}$ (green ellipses).}
\label{fig:illustrate}
\end{figure}
The following theorem establishes recursive feasibility and robust constraint satisfaction of the proposed nonlinear RAMPC framework. 
\begin{theorem}
\label{thm:main}
Let Assumptions~\ref{ass:model}, \ref{ass:nominal}, \ref{ass:increm}, \ref{ass:w_tilde} and \ref{ass:term} hold.
Suppose that Problem~\eqref{eq:RAMPC} is feasible at $t=0$. 
Then Problem~\eqref{eq:RAMPC} is recursively feasible and the constraints~\eqref{eq:constraint_condition} are satisfied for the resulting closed-loop system. 
\end{theorem}
\begin{proof}
The following proof is an extension of~\cite[Thm.~1]{Robust_TAC_19} to recursively updated nominal parameters $\overline{\theta}_t$ and uncertainty sets $\widetilde{\Theta}_t$. 
We first construct a suitable candidate solution based on the stabilizability condition (Ass.~\ref{ass:increm}) and derive a bound on the deviation between the candidate solution and the previous optimal solution. 
Then, as the main step, we show a nestedness property between the previous optimal solution and the new candidate solution, compare Fig.~\ref{fig:illustrate} for an illustration using ellipsoidal sets.
 Finally, we show that the candidate solution also satisfies the posed inequality constraints~\eqref{eq:RAMPC_w}, \eqref{eq:RAMPC_con}, \eqref{eq:RAMPC_con_sw}, \eqref{eq:RAMPC_term}.  \\
\textbf{Part I. } Candidate solution:
For convenience, define
\begin{align}
\label{eq:extended_solution}
&\overline{u}^*_{N|t}=k_f(\overline{x}^*_{N|t}),~w^*_{N|t}=\tilde{w}_{\delta,\widetilde{\Theta}_{t},\mathbb{D}}(\overline{x}^*_{N|t},\overline{u}^*_{N|t},s_{N|t}^*),\\
&\overline{x}^*_{N+1|t}=f_{\overline{\theta}_t}(\overline{x}^*_{N|t},\overline{u}_{N|t}),~s^*_{N+1|t}=\rho_{\overline{\theta}_t}s^*_{N|t}+w^*_{N|t}.\nonumber
\end{align}
As a candidate solution, we use the stabilizing feedback $\kappa$ to stabilize the previous optimal solution, i.e.,  
\begin{subequations}
\label{eq:candidate}
\begin{align}
\label{eq:candidate_1}
\overline{u}_{k|t+1}=&\kappa(\overline{x}_{k|t+1},\overline{x}^*_{k+1|t},\overline{u}^*_{k+1|t}),\\
\label{eq:candidate_2}
w_{k|t+1}=&\tilde{w}_{\delta,\widetilde{\Theta}_{t+1},\mathbb{D}}(\overline{x}_{k|t+1},\overline{u}_{k|t+1},s_{k|t+1}),
\end{align}
\end{subequations}
for $k=0,\dots,N-1$, with $\overline{x},\hat{x},\hat{u},s$ defined according to~\eqref{eq:RAMPC_init}, \eqref{eq:RAMPC_dyn}, \eqref{eq:RAMPC_dyn_s} with $\overline{\theta}_{t+1}$, $\widetilde{\Theta}_{t+1}$, $\hat{\theta}_{t+1}$. 
Note that, due to the parameter change $\Delta\overline{\theta}_t$ the prediction model $f_{\overline{\theta}}$ changes, yielding:  
\begin{align*}
\overline{x}_{k+1|t+1}=&f_{\overline{\theta}_{t+1}}(\overline{x}_{k|t+1},\overline{u}_{k|t+1})\\
\stackrel{\eqref{eq:model_affine}}=&f_{\overline{\theta}_t}(\overline{x}_{k|t+1},\overline{u}_{k|t+1})+G(\overline{x}_{k|t+1},\overline{u}_{k|t+1})\Delta\overline{\theta}_t.
\end{align*}
Let us define 
\begin{align}
\label{eq:def_tilde_s}
\tilde{s}_{0|t+1}:=&s^*_{1|t},\\
\tilde{s}_{k+1|t+1}:=&\rho_{\overline{\theta}_{t}}\tilde{s}_{k|t+1}+\tilde{w}_{\delta,\Delta\widetilde{\Theta},\{0\}}(\overline{x}^*_{k+1|t},\overline{u}^*_{k+1|t},\tilde{s}_{k|t+1}).\nonumber
\end{align}
In the following, we show that  $\tilde{s}_{\cdot|t+1}$ bounds the deviation between the previous optimal trajectory $\overline{x}^*_{\cdot|t}$ and the candidate solution $\overline{x}_{\cdot|t+1}$, i.e., we show the following inequality by induction:
\begin{align}
\label{eq:bound_candidate}
V_{\delta}(\overline{x}_{k|t+1},\overline{x}^*_{k+1|t})\leq \tilde{s}_{k|t+1},~k=0,\dots,N.
\end{align}
Induction start: Condition~\eqref{eq:bound_candidate} is satisfied at $k=0$ with
\begin{align}
\label{eq:bound_candidate_induction_start}
&V_{\delta}(\overline{x}_{0|t+1},\overline{x}^*_{1|t})\nonumber\\
\stackrel{\eqref{eq:model_affine},\eqref{eq:RAMPC_dyn}}{=}&V_{\delta}(f_{\overline{\theta}_t}(x_t,u_t)+Ed_t+G(x_t,u_t)(\theta^*-\theta_t),f_{\overline{\theta}_t}(x_t,u_t)) \nonumber\\
\stackrel{\eqref{eq:model_mismatch},\eqref{eq:w_tilde_a}}{\leq}& \tilde{w}_{\widetilde{\Theta}_t,\mathbb{D}}(x_t,u_t)\stackrel{\eqref{eq:RAMPC_w}}{\leq} w^*_{0|t}\stackrel{\eqref{eq:RAMPC_dyn_s}}{=}s^*_{1|t}\stackrel{\eqref{eq:def_tilde_s}}{=}\tilde{s}_{0|t+1}. 
\end{align}
Induction step: Suppose~\eqref{eq:bound_candidate} holds for some $k\in\{0,\dots,N-1\}$, then condition~\eqref{eq:bound_candidate} also holds at $k+1$ using
\begin{align*}
&V_{\delta}(\overline{x}_{k+1|t+1},\overline{x}^*_{k+2|t})\\
\stackrel{\eqref{eq:RAMPC_dyn}}{=}&V_{\delta}(f_{\overline{\theta}_{t+1}}(\overline{x}_{k|t+1},\overline{u}_{k|t+1}),f_{\overline{\theta}_t}(\overline{x}^*_{k+1|t},\overline{u}^*_{k+1|t})) \nonumber\\
\stackrel{\eqref{eq:increm_d}}{\leq} &V_{\delta}(f_{\overline{\theta}_{t}}(\overline{x}_{k|t+1},\overline{u}_{k|t+1}),f_{\overline{\theta}_t}(\overline{x}^*_{k+1|t},\overline{u}^*_{k+1|t}))\\
&+V_{\delta}(\overline{x}^*_{k+2|t}+G(\overline{x}_{k|t+1},\overline{u}_{k|t+1})\Delta\overline{\theta}_t,\overline{x}^*_{k+2|t})\\
\stackrel{ \eqref{eq:increm_c},\eqref{eq:w_tilde_a}}{\leq} &\rho_{\overline{\theta}_t}V_{\delta}(\overline{x}_{k|t+1},\overline{x}^*_{k+1|t})
+\tilde{w}_{\Delta\widetilde{\Theta}_t,\{0\}}(\overline{x}_{k|t+1},\overline{u}_{k|t+1})\\
\stackrel{\eqref{eq:w_tilde_b},\eqref{eq:candidate_1}}{\leq} &(\rho_{\overline{\theta}_t}+L_{\Delta\widetilde{\Theta}})V_{\delta}(\overline{x}_{k|t+1},\overline{x}^*_{k+1|t})\\
&+\tilde{w}_{\Delta\widetilde{\Theta}_t,\{0\}}(\overline{x}^*_{k+1|t},\overline{u}^*_{k+1|t})\\
\stackrel{\eqref{eq:bound_candidate}}{\leq} &\rho_{\overline{\theta}_t}\tilde{s}_{k|t+1}+\tilde{w}_{\delta,\Delta\widetilde{\Theta}_t,\{0\}}(\overline{x}^*_{k+1|t},\overline{u}^*_{k+1|t},\tilde{s}_{k|t+1})\\
\stackrel{\eqref{eq:def_tilde_s}}{=}&\tilde{s}_{k+1|t+1}. 
\end{align*}
The bound $\tilde{s}$ consists of two components: first, a term bounding the initial prediction mismatch $x_{t+1}-\overline{x}^*_{1|t}$ using $w^*_{0|t}$, and a  second term $\tilde{w}_{\Delta\widetilde{\Theta}_t,\{0\}}$, which depends on the parameter update. 
In the absence of parameter updates, we recover the robust MPC proof in~\cite{Robust_TAC_19} as a special case with $\tilde{s}_{k|t+1}=\rho^kw^*_{0|t}$. \\
\textbf{Part II. } In the following we show that the new candidate solution $\overline{x}$ with corresponding tube $s$ satisfies a nestedness property (c.f. Fig~\ref{fig:illustrate}) w.r.t. the previous optimal solution. 
First, note that the following bound holds
\begin{align}
\label{eq:w_bound_alternative}
&w_{k|t+1}
\stackrel{\eqref{eq:candidate_2}}{=} 
\tilde{w}_{\delta,\widetilde{\Theta}_{t+1},\mathbb{D}}(\overline{x}_{k|t+1},\overline{u}_{k|t+1},s_{k|t+1})\nonumber\\
\stackrel{\eqref{eq:w_tilde_g}, \eqref{eq:candidate_1},\eqref{eq:bound_candidate}}{\leq} &\tilde{w}_{\delta,\widetilde{\Theta}_{t+1},\mathbb{D}}(\overline{x}^*_{k+1|t},\overline{u}^*_{k+1|t},s_{k|t+1}+\tilde{s}_{k|t+1})\nonumber\\
\stackrel{\eqref{eq:w_tilde_c_delta}}{\leq }&\tilde{w}_{\delta,\widetilde{\Theta}_{t},\mathbb{D}}(\overline{x}^*_{k+1|t},\overline{u}^*_{k+1|t},s_{k|t+1}+\tilde{s}_{k|t+1})\nonumber\\
&-\tilde{w}_{\delta,\Delta\widetilde{\Theta}_{t},\{0\}}(\overline{x}^*_{k+1|t},\overline{u}^*_{k+1|t},s_{k|t+1}+\tilde{s}_{k|t+1})\nonumber\\
\stackrel{\eqref{eq:def_w_delta_tilde}, \eqref{eq:RAMPC_w}}{\leq}&w^*_{k+1|t}+L_{\widetilde{\Theta}_t}(s_{k|t+1}-s_{k+1|t}^*+\tilde{s}_{k|t+1})\\
&-\tilde{w}_{\delta,\Delta\widetilde{\Theta}_{t},\{0\}}(\overline{x}^*_{k+1|t},\overline{u}^*_{k+1|t},s_{k|t+1}+\tilde{s}_{k|t+1}).\nonumber
\end{align}
In the following, we show
\begin{align}
\label{eq:s_bound_alternative}
&s_{k|t+1}+\tilde{s}_{k|t+1}-s^*_{k+1|t}\leq 0,
\end{align}
for $k=0,\dots,N$, using a proof by induction. 
Induction start: Condition~\eqref{eq:s_bound_alternative} is satisfied at $k=0$ with equality:
\begin{align*}
s_{0|t+1}-s^*_{1|t}+\tilde{s}_{0|t+1}\stackrel{\eqref{eq:RAMPC_init},\eqref{eq:RAMPC_dyn_s},\eqref{eq:def_tilde_s}}{=}0-w^*_{0|t}+w^*_{0|t}=0.
\end{align*}
Induction step: Suppose~\eqref{eq:s_bound_alternative} holds for some $k\in\{0,\dots,N-1\}$, then condition~\eqref{eq:s_bound_alternative} holds at $k+1$ using
\begin{align*}
&s_{k+1|t+1}+\tilde{s}_{k+1|t+1}-s^*_{k+2|t}\\
\stackrel{\eqref{eq:RAMPC_dyn_s},\eqref{eq:def_tilde_s}}{=}&\rho_{\overline{\theta}_{t+1}}s_{k|t+1}+w_{k|t+1}+\rho_{\overline{\theta}_t}(\tilde{s}_{k|t+1}-s^*_{k+1|t})\\
&+\tilde{w}_{\delta,\Delta\widetilde{\Theta}_t,\{0\}}(\overline{x}^*_{k+1|},\overline{u}^*_{k+1|t},\tilde{s}_{k|t+1})-w^*_{k+1|t}\\
\stackrel{\eqref{eq:rho_Lipschitz},\eqref{eq:w_tilde_e}}{\leq} &(\rho_{\overline{\theta}_t}+L_{\Delta\widetilde{\Theta}_t})(s_{k|t+1}+\tilde{s}_{k|t+1}-s^*_{k+1|t})+w_{k|t+1}\\
&-w^*_{k+1|t}
+\tilde{w}_{\delta,\Delta\widetilde{\Theta}_t,\{0\}}(\overline{x}^*_{k+1|},\overline{u}^*_{k+1|t},s^*_{k+1|t}) \\
\stackrel{\eqref{eq:w_bound_alternative}}{\leq}&(\rho_{\overline{\theta}_t}+L_{\widetilde{\Theta}_t})(s_{k|t+1}+\tilde{s}_{k|t+1}-s^*_{k+1|t})\stackrel{\eqref{eq:s_bound_alternative}}{\leq} 0.
\end{align*}
\textbf{Part III. } Constraints~\eqref{eq:RAMPC_con}, \eqref{eq:RAMPC_con_sw}, \eqref{eq:RAMPC_term}:
Regarding the terminal set~\eqref{eq:RAMPC_term} constraint: 
First, note that we have
\begin{align}
\label{eq:w_intermediate_bound_2}
\tilde{w}_{\delta,\widetilde{\Theta}_t,\mathbb{D}}(\overline{x}^*_{k+1|t},\overline{u}^*_{k+1|t},s^*_{k+1|t})\stackrel{\eqref{eq:RAMPC_w}}{\leq}  w_{k+1|t}^*\stackrel{\eqref{eq:RAMPC_con_sw}}{\leq} \overline{w}_{\widetilde{\Theta}_t},
\end{align}
for $k=0,\dots,N-1$, with $k=N-1$ using~\eqref{eq:term_3}, which implies  
\begin{align*}
\overline{w}_{\Delta\widetilde{\Theta}_t}\stackrel{\eqref{eq:term_6}}{\geq } \tilde{w}_{\delta,\Delta\widetilde{\Theta}_t,\mathbb{D}}(\overline{x}_{k+1|t}^*,\overline{u}_{k+1|t}^*,s^*_{k+1|t})\\
\stackrel{\eqref{eq:w_tilde_c_delta}}{\geq} \tilde{w}_{\delta,\Delta\widetilde{\Theta}_t,\{0\}}(\overline{x}_{k+1|t}^*,\overline{u}_{k+1|t}^*,s^*_{k+1|t})+\overline{d}. 
\end{align*}
Thus, \eqref{eq:def_tilde_s} ensures 
\begin{align*}
&(\rho_{\overline{\theta}_{t}}+L_{\Delta\widetilde{\Theta}_t})^N\overline{d}\leq \tilde{s}_{N|t+1}\\
\leq&(\rho_{\overline{\theta}_{t}}+L_{\Delta\widetilde{\Theta}_t})^N\overline{w}_{\widetilde{\Theta}_t}+(\overline{w}_{\Delta\widetilde{\Theta}_t}-\overline{d})\sum_{k=0}^{N-1}(\rho_{\overline{\theta}_t}+L_{\Delta\widetilde{\Theta}_t})^k.
\end{align*}
Furthermore, we have $V_{\delta}(\overline{x}_{N|t+1},\overline{x}^*_{N+1|t})\stackrel{\eqref{eq:bound_candidate}}{\leq} \tilde{s}_{N|t+1}$, $s_{N|t+1}\stackrel{\eqref{eq:s_bound_alternative}}{\leq} \rho_{\overline{\theta}_t}s^*_{N|t}+\tilde{w}_{\delta,\widetilde{\Theta}_t,\mathbb{D}}(\overline{x}^*_{N|t},k_f(\overline{x}^*_{N|t}),s^*_{N|t})-\tilde{s}_{N|t+1}$.
Thus, $(\overline{x}^*_{N|t},{s}^*_{N|t})\in\mathcal{X}_{f,\overline{\theta}_t,\widetilde{\Theta}_t}$
 ensures $(\overline{x}_{N|t+1},s_{N|t+1})\in\mathcal{X}_{f,\overline{\theta}_{t+1},\widetilde{\Theta}_{t+1}}$ using~\eqref{eq:term_1}. 
Satisfaction of the tightened state and input constraints~\eqref{eq:RAMPC_con} direclty follows from \eqref{eq:bound_candidate}, \eqref{eq:s_bound_alternative} 
\begin{align*}
&h_j(\overline{x}_{k|t+1},\overline{u}_{k|t+1})+c_js_{k|t+1}\\
\stackrel{\eqref{eq:c_j},\eqref{eq:candidate_1},\eqref{eq:bound_candidate}}{\leq} &h_j(\overline{x}^*_{k+1|t},\overline{u}^*_{k+1|t})+c_j\tilde{s}_{k|t+1}+c_js_{k|t+1}\\
\stackrel{\eqref{eq:s_bound_alternative}}{\leq}& h_j(\overline{x}^*_{k+1|t},\overline{u}^*_{k+1|t})+c_js^*_{k+1|t}\stackrel{\eqref{eq:RAMPC_con}}{\leq} 0.
\end{align*} 
for $k=0,\dots,N-2$. 
Satisfaction at $k=N-1$ follows similarly, using \eqref{eq:term_2} and~\eqref{eq:candidate_1}.
The constraints $s_{k|t+1}\leq \overline{s}$ in~\eqref{eq:RAMPC_con_sw} hold  due to~\eqref{eq:s_bound_alternative} for $k=0,\dots,N-2$ and using~\eqref{eq:term_4} for $k=N-1$. 
Using $\widetilde{\Theta}_{t+1}\subseteq\widetilde{\Theta}_t$, condition~\eqref{eq:w_intermediate_bound_2} implies
\begin{align*}
w_{k|t+1}\stackrel{\eqref{eq:w_bound_alternative},\eqref{eq:s_bound_alternative}}{\leq} \tilde{w}_{\delta,\widetilde{\Theta}_{t+1},\mathbb{D}}(\overline{x}^*_{k+1|t},\overline{u}_{k+1|t}^*,s^*_{k+1|t})\stackrel{\eqref{eq:w_intermediate_bound_2},\eqref{eq:term_6}}{\leq} \overline{w}_{\widetilde{\Theta}_{t+1}},
\end{align*}
for $k=0,\dots,N-1$ and thus satisfaction~\eqref{eq:RAMPC_con_sw}.

We would like to point out that we repeatedly used $(\overline{x}_{k|t+1},\overline{u}_{k|t+1})\in\mathcal{Z}$ and $V_{\delta}(\overline{x}_{k|t+1},\overline{x}_{k+1|t})\leq \delta_{loc}$ when applying conditions from Assumption~\ref{ass:increm}, which holds since $V_{\delta}$ is continuous, $0\leq \tilde{s}_{k|t+1}\leq \tilde{s}_{k|t+1}+s_{k|t+1}\stackrel{\eqref{eq:s_bound_alternative}}{\leq} s^*_{k+1|t}\stackrel{\eqref{eq:RAMPC_con_sw}}{\leq} \overline{s}\leq \delta_{loc}$ and \eqref{eq:RAMPC_con} holds.
\end{proof}
 In~\cite{Robust_TAC_19}, it was shown that a similar online constructed tube size $s_{\cdot|t}$  in combination with the constraint tightening~\eqref{eq:RAMPC_con} ensures robust recursive feasibility and robust constraint satisfaction. 
Theorem~\ref{thm:main} extends this robust MPC framework to adaptive MPC, utilizing online updates of the nominal model and the uncertainty to reduce conservatism. 
The approach is somewhat similar to the RAMPC approaches in~\cite{adetola2011robust,guay2015robust,gonccalves2016robust}, which also uses some scalar dynamics in $s$ to ensure robust constraint satisfaction.
However, the proposed framework is significantly more flexible, compare the designs in Sec.~\ref{sec:design_1}--\ref{sec:special_guay} and the numerical example in Sec.~\ref{sec:num}.

%% file: Theory_4.tex
%!TEX root = ./Adaptive_Nonlin.tex
%%%%%%%%%%%%%%%%%%%%%%%%%%%%%%%%%%%%%%%%%%%%%%%%%%%%%%%%%%%%%%%%%%%%%%%%%%%%%%%
\subsection{Stability results and LMS updates} 
\label{sec:main_4}
In the following, we discuss the LMS update and show finite-gain stability w.r.t. additive disturbances $d_t$. 
Before proceedings, we would like to point out that simply using $\hat{\theta}_t=\overline{\theta}_t$ would be sufficient to show weaker practical asymptotic stability properties similar to the stability properties of the robust MPC in~\cite{Robust_TAC_19}. 
Denote the prediction error $\tilde{x}_{1|t}=x_{t+1}-f_{\hat{\theta}_t}(x_t,u_t)$. 
The LMS point estimate update is given by 
\begin{subequations}
\label{eq:hat_theta_update}
\begin{align}
\label{eq:hat_theta_update_1}
\tilde{\theta}_t:=&\hat{\theta}_{t-1}+\mu G(x_{t-1},u_{t-1})^\top \tilde{x}_{1|t-1},\\
\label{eq:hat_theta_update_2}
\hat{\theta}_t:=&\arg\min_{\theta\in\overline{\theta}_t\oplus\widetilde{\Theta}_t}\|\theta- \tilde{\theta}_{t} \|,
\end{align}
with some initial parameter $\hat{\theta}_0\in\overline{\theta}_0\oplus\widetilde{\Theta}_0$.
The update~\eqref{eq:hat_theta_update_1} corresponds to a standard LMS update, while the projection step in~\eqref{eq:hat_theta_update_2} uses the existing set membership bound to improve the estimate, which is also necessary to ensure the desired properties. 
Similar LMS updates have also been considered in the linear RAMPC schemes in~\cite{lorenzen2019robust,Koehler2019Adaptive,bujarbaruah2019adaptive}.
The update gain $\mu>0$ is chosen such that
\begin{align}
\label{eq:mu}
\dfrac{1}{\mu}>\|G(x,u)\|^2,~\forall (x,u)\in\mathcal{Z},
\end{align}
\end{subequations}
which is possible since $\mathcal{Z}$ compact and $G$ is Lipschitz continuous.

The following analysis is an extension of the stability proof in~\cite[Thm.~14]{lorenzen2019robust} to nonlinear systems.
\begin{theorem}
\label{thm:stability}
Suppose the conditions in Theorem~\ref{thm:main} hold. 
Assume further that $\ell(x,u)=\|x\|_Q^2+\|u\|_R^2$, $V_f(x)=\|x\|_{P_f}^2$  with $Q,R,P_f$ positive definite, the feedback $\kappa$ in Assumption~\ref{ass:increm} has the form $\kappa(x,z,v)=v+\kappa_x(x)-\kappa_x(z)$ and the terminal set $\mathcal{X}_f$ has a non-empty interior.
Then the closed-loop system~\eqref{eq:close} is finite-gain $\mathcal{L}_2$ stable w.r.t. additive disturbances $d_t$, i.e., there exist constants $c_0,c_1,c_2>0$, such that for all $T\in\mathbb{N}$:
\begin{align}
\label{eq:finite_gain_stable}
\sum_{t=0}^{T}\|x_t\|^2\leq c_0\|x_0\|^2+c_1\|\Hat{\theta}_0-\theta^*\|^2+c_2\sum_{t=0}^{T}\|d_t\|^2.
\end{align} 
\end{theorem}
\begin{proof}
\textbf{Part I. } 
First, we show that any feasible solution in~\eqref{eq:RAMPC} satisfies $(\hat{x}_{k|t},\hat{u}_{k|t})\in\mathcal{Z}$.   
Similar to~\eqref{eq:bound_candidate}, one can show $V_{\delta}(\hat{x}_{k|t},\overline{x}_{k|t})\leq s_{k|t}$ with a proof of induction using: 
\begin{align*}
&V_{\delta}(\hat{x}_{k+1|t},\overline{x}_{k+1|t})\\
\stackrel{\eqref{eq:RAMPC_dyn},\eqref{eq:increm_d}}{\leq} &V_{\delta}(f_{\overline{\theta}_t}(\hat{x}_{k|t},\hat{u}_{k|t}),f_{\overline{\theta}_t}(\overline{x}_{k|t},\overline{u}_{k|t}))\\
&+V_{\delta}(\overline{x}_{k+1|t}+G(\hat{x}_{k|t},\hat{u}_{k|t})(\hat{\theta}_t-\overline{\theta}_t),\overline{x}_{k+1|t}) \\
\stackrel{\eqref{eq:RAMPC_u_hat},\eqref{eq:increm_c},\eqref{eq:w_tilde_a}}{\leq} &\rho_{\overline{\theta}_t}V_{\delta}(\hat{x}_{k|t},\overline{x}_{k|t})+\tilde{w}_{\hat{\theta}_t-\overline{\theta}_t,\{0\}}(\hat{x}_{k|t},\hat{u}_{k|t})\\
\stackrel{\eqref{eq:def_w_delta_tilde}, \eqref{eq:w_tilde_g}, \eqref{eq:w_tilde_c_delta},\eqref{eq:hat_theta_update_2}}{\leq} &\rho_{\overline{\theta}_t}s_{k|t}+\tilde{w}_{\delta,\widetilde{\Theta}_t,\mathbb{D}}(\overline{x}_{k|t},\overline{u}_{k|t},s_{k|t})\\
\stackrel{\eqref{eq:RAMPC_dyn_s},\eqref{eq:RAMPC_w}}{\leq}& s_{k+1|t}.
\end{align*}
Constraint satisfaction follows directly using 
\begin{align}
\label{eq:constraint_satisfaction_hat_x}
h_j(\hat{x}_{k|t},\hat{u}_{k|t})\stackrel{\eqref{eq:c_j}}\leq h_j(\overline{x}_{k|t},\overline{u}_{k|t})+c_js_{k|t}\stackrel{\eqref{eq:RAMPC_con}}{\leq} 0.
\end{align}
\textbf{Part II. } 
Using the assumed structure of $\kappa$, the candidate input $\hat{u}$ satisfies
 \begin{align}
 \label{eq:hat_u_candidate}
&\hat{u}_{k|t+1}\stackrel{\eqref{eq:RAMPC_u_hat}}{=}\kappa(\hat{x}_{k|t+1},\overline{x}_{k|t+1},\overline{u}_{k|t+1})\\
\stackrel{\eqref{eq:candidate_1}}{=}&\kappa(\hat{x}_{k|t+1},\overline{x}^*_{k+1|t},\overline{u}^*_{k+1|t})
\stackrel{\eqref{eq:RAMPC_u_hat}}{=}\kappa(\hat{x}_{k|t+1},\hat{x}^*_{k+1|t},\hat{u}^*_{k+1|t}),\nonumber
\end{align}
for $k=0,\dots,N-1$.
Denote $\Delta\hat{\theta}_t=\hat{\theta}_{t+1}-\hat{\theta}_t$. We have
\begin{align*}
&V_{\delta}(\hat{x}_{k+1|t+1},\hat{x}^*_{k+2|t})\\
=&V_{\delta}(f_{\hat{\theta}_{t+1}}(\hat{x}_{k|t+1},\hat{u}_{k|t+1}),f_{\hat{\theta}_t}(\hat{x}^*_{k+1|t},\hat{u}^*_{k+1|t}))\\
\stackrel{\eqref{eq:increm_d}}{\leq} &V_{\delta}(f_{\hat{\theta}_{t}}(\hat{x}_{k|t+1},\hat{u}_{k|t+1}),f_{\hat{\theta}_t}(\hat{x}^*_{k+1|t},\hat{u}^*_{k+1|t}))\\
&+V_{\delta}(\hat{x}^*_{k+2|t}+G(\hat{x}_{k+1|t+1},\hat{u}_{k+1|t+1})\Delta\hat{\theta}_t,\hat{x}^*_{k+2|t})\\
\stackrel{\eqref{eq:increm_a},\eqref{eq:increm_c},\eqref{eq:hat_u_candidate}}{\leq} &\rho_{\hat{\theta}_t}V_{\delta}(\hat{x}_{k|t+1},\hat{x}^*_{k+1|t})\\
&+c_{\delta,u}\|G(\hat{x}_{k|t+1},\hat{u}_{k|t+1})\|\|\Delta\hat{\theta}_t\|
\end{align*}
Due to the fact that the projection operator is non-expansive, we have $\|\Delta\hat{\theta}_t\|\leq \|\tilde{\theta}_{t+1}-\hat{\theta}_t\|=\mu\|G(x_t,u_t)^\top \tilde{x}_{1|t}\|$.
Using this bound, constraint satisfaction~\eqref{eq:constraint_satisfaction_hat_x} and the choice of $\mu$~\eqref{eq:mu}, we obtain
\begin{align*}
V_{\delta}(\hat{x}_{k+1|t+1},\hat{x}^*_{k+2|t})\leq &\rho_{\hat{\theta}_t}V_{\delta}(\hat{x}_{k|t+1},\hat{x}^*_{k+1|t})+c_{\delta,u}\|\tilde{x}_{1|t}\|,
\end{align*}
which recursively applied
 ensures
\begin{align*}
V_{\delta}(\hat{x}_{k|t+1},\hat{x}^*_{k+1|t})\leq &c_{\delta,u}\|\tilde{x}_{1|t}\|\sum_{j=0}^{k}\rho_{\hat{\theta}_t}^j\stackrel{\eqref{eq:rho_Lipschitz},\eqref{eq:w_tilde_e},\eqref{eq:hat_theta_update_2}}{\leq} c_1 \|\tilde{x}_{1|t}\|,\\
c_1:=&c_{\delta,u}\sum_{j=0}^{N}(\rho_{\overline{\theta}_0}+L_{\widetilde{\Theta}_0})^j .\nonumber
\end{align*}
for $k=0,\dots,N$ using $\rho_{\hat{\theta}_t}\leq \rho_{\overline{\theta}_0}+L_{\widetilde{\Theta}_0}$  and initial condition $V_{\delta}(\hat{x}_{0|t+1},\hat{x}^*_{1|t})\leq c_{\delta,u}\|\tilde{x}_{1|t}\|$.
Correspondingly, the deviation in state $\Delta\hat{x}_{k|t+1}:=\hat{x}_{k|t+1}-\hat{x}^*_{k+1|t}$and  input $\Delta \hat{u}_{k|t+1}:=\hat{u}_{k|t+1}-\hat{u}^*_{k+1|t}$ satisfy
\begin{align*}
\|\Delta \hat{x}_{k|t+1}\|\stackrel{\eqref{eq:increm_a}}{\leq}& V_{\delta}(\hat{x}_{k|t+1},\hat{x}^*_{k+1|t})/c_{\delta,l}\leq c_1/c_{\delta,l}\|\tilde{x}_{1|t}\|,\\
\|\Delta\hat{u}_{k|t+1}\|\stackrel{\eqref{eq:increm_b},\eqref{eq:hat_u_candidate}}{\leq} &\kappa_{\max} V_{\delta}(\hat{x}_{k|t+1},\hat{x}^*_{k+1|t})\leq \kappa_{\max}c_1\|\tilde{x}_{1|t}\|.
\end{align*}
Using the Cauchy-Schwarz and Young's inequality and the quadratic stage cost $\ell$ this implies for any $\epsilon>0$:
\begin{align*}
\ell(\hat{x}_{k|t+1},\hat{u}_{k|t+1})\leq &(1+\epsilon)\ell(\hat{x}^*_{k+1|t},\hat{u}^*_{k+1|t})+\left(1+\frac{1}{\epsilon}\right)c_2\|\tilde{x}_{1|t}\|^2,\nonumber\\
{c}_2:=&c_1^2(\lambda_{\max}(Q)/c_{\delta,l}^2+\lambda_{\max}(R)\kappa_{\max}^2).
\end{align*}
Similarly, the terminal cost $V_f$ satisfies
\begin{align*}
V_f(\hat{x}_{N|t+1})\leq &(1+\epsilon)V_f(\hat{x}^*_{N+1|t})+\left(1+\frac{1}{\epsilon}\right)c_3\|\tilde{x}_{1|t}\|^2,\nonumber\\
{c}_3:=&c_1^2\lambda_{\max}(P_f)/c_{\delta,l}^2.
\end{align*}
Finally, we note that the quadratic terminal cost $V_f$, compact constraints $\mathcal{Z}$, and the fact that the terminal set has a non-empty interior  ensures that there exists some constant $c_4>0$ such that the value function $V_t$ satisfies
$V_t\leq c_4 \|x_t\|^2$ for any feasible point $x_t$, compare~\cite[Prop. 2.18]{rawlings2012postface}.
The remainder of the proof is analogous to the linear case in~\cite[Thm.~14]{lorenzen2019robust}, using a small enough $\epsilon>0$ and the following bound of the LMS (c.f.~\cite[Lemma~5]{lorenzen2019robust})
\begin{align}
\label{eq:intermediate_LMS}
\sum_{t=0}^{T}\|\tilde{x}_{1|t}\|^2\leq \dfrac{1}{\mu}\|\hat{\theta}_0-\theta^*\|^2+\sum_{t=0}^T\|d_t\|^2.
\end{align} 
\end{proof}
Considering a quadratic stage cost $\ell$, a terminal cost $V_f$ and a terminal set $\mathcal{X}_f$ with a non-empty interior is quite common in MPC, compare e.g.~\cite{rawlings2009model,chen1998quasi}. 
The considered restriction for $\kappa$  is equivalent to assuming that some nonlinear feedback $\kappa_x(x)$ exists, such that the system is incrementally stable.  
While most formulations in the robust MPC literature satisfy the conditions using linear feedbacks $K$~\cite{lorenzen2019robust,Lu2019RAMPC,Koehler2019Adaptive,hewing2019cautious,mckinnon2019learn,limon2005robust,bayer2013discrete}, quasi-LPV based designs can only be applied if the feedback $K(z,v)$ is not parametrized by the input $v$ (c.f.~\cite[Prop.~3]{JK_QINF}), control contraction metrics~\cite{singh2017robust} only satisfy this conditions for linear feedbacks $K$, and boundary layer controllers~\cite{lopez2018adaptive,DynamicTube_Lopez_19}, \cite[Cor.~6]{villanueva2017robust} typically do not satisfy this condition.

%% file: Approaches_1.tex
%!TEX root = ./Adaptive_Nonlin.tex
%%%%%%%%%%%%%%%%%%%%%%%%%%%%%%%%%%%%%%%%%%%%%%%%%%%%%%%%%%%%%%%%%%%%%%%%%%%%%%%
\subsection{Design procedure and overall algorithm}
\label{sec:design_1}
In the following we provide explicit design procedures satisfying the conditions in  Assumptions~\ref{ass:nominal}, \ref{ass:increm},  \ref{ass:w_tilde}, \ref{ass:term}.
Here the focus is on providing procedures that are simple to apply. 
We first discuss set membership estimation (Ass.~\ref{ass:nominal}) using Alg.~\ref{alg:HC} and Lemma~\ref{lemma:setmembership}.
Then we discuss the design of the incremental Lyapunov function $V_{\delta}$ (Ass.~\ref{ass:increm}), and the computation of $L_{\Delta\Theta,\rho}$ (Prop.~\ref{prop:properties_2_quadratic}).
Then we provide two different design for the uncertainty characterization $\tilde{w}_\delta$ (Ass.~\ref{ass:w_tilde}) in Prop.~\ref{prop:design_w_1} and \ref{prop:design_w_2}. 
Prop.~\ref{prop:terminal_design} provides a simple design for the terminal ingredients (Ass.~\ref{ass:term}).
The overall offline design and online operation are summarized in Algorithms~\ref{alg:offline} and \ref{alg:online}.
\subsubsection*{Set membership estimation} 
In the following, we detail how $\overline{\theta}_t,\tilde{\Theta}_t$ satisfying Assumption~\ref{ass:nominal} can be computed using set membership estimation, similar to~\cite{tanaskovic2014adaptive,lorenzen2019robust,Lu2019RAMPC,Koehler2019Adaptive,lopez2018adaptive,bujarbaruah2019adaptive}. 
Given $(x_{t-1},u_{t-1},x_t)$, the non falsified parameter set is given by the polytope 
\begin{align*}
\Delta_t:=\{\theta\in\mathbb{R}^p|~x_t-f(x_{t-1},u_{t-1}) -G(x_{t-1},u_{t-1})\theta\in\mathbb{D}\}.
\end{align*}
For $t\leq 0$ set $\Delta_{k}=\mathbb{R}^p$. 
For simplicity, we consider hypercubes of the form $\widetilde{\Theta}_{t-1}:=\eta_{t-1}\mathbb{B}_{\infty}$ with some scalar $\eta_{t-1}\geq 0$.
The following algorithm uses a given hypercube $\overline{\theta}_{t-1}\oplus\eta_{t-1}\mathbb{B}_{\infty}$ and the past $M\in\mathbb{N}$ sets $\Delta_{t-k}$ in a moving window fashion to compute a smaller uncertainty set. 
\begin{algorithm}[h]
\caption{Moving window set membership updates} 
\label{alg:HC}
Input: $\{\Delta_{k}\}_{k=t,\dots,t-M-1}$, $\overline{\theta}_{t-1},\eta_{t-1}$. Output: $\overline{\theta}_t$, $\eta_t$
\begin{algorithmic}
\State Define set $\Theta_t^M:=(\overline{\theta}_{t-1}\oplus\eta_{t-1}\mathbb{B}_{\infty})\bigcap_{k=t-M-1}^{t}\Delta_k$.
\State Solve $2 p$ optimization problems ($i=1,\dots,p$):    
\Statex ~$\theta_{i,t,\min}:=\min_{\theta\in\Theta_t^M}e_i^\top \theta$, $\theta_{i,t,\max}:=\max_{\theta\in\Theta_t^M}e_i^\top \theta$,
\Statex ~\text{with unit vector }$e_i=[0,\dots,1,\dots,0]\in\mathbb{R}^p,~[e_i]_i=1$.
\State Set $[\overline{\theta}_t]_i=0.5(\theta_{i,t,\min}+\theta_{i,t,\max})$.
\State Set $\eta_t=0.5\max_i(\theta_{i,t,\max}-\theta_{i,t,\min})$.
\State Project: $\overline{\theta}_t$ on $\overline{\theta}_{t-1}\oplus (\eta_{t-1}-\eta_t)\mathbb{B}_{\infty}$.
\end{algorithmic}
\end{algorithm}

This algorithm first computes the \textit{unique} hyperbox overapproximation ($\theta_{i,t,\min},\theta_{i,t,\max}$) of the non-falsified parameter set $\Theta_t^M$. 
Then in a second step an overapproximating hypercube in form of $\overline{\theta}_t\oplus\eta_t\mathbb{B}_{\infty}$ is computed. 
In case that $\mathbb{D}$ polytoptic, $\Theta_t^M$ is a polytope and  executing Alg.~\ref{alg:HC} requires the solution to $2p$ linear programs (LPs) with $p$ decision variables. 
\begin{lemma}
\label{lemma:setmembership}
Let Assumptions~\ref{ass:model} hold and suppose $\theta^*\in\overline{\theta}_0\oplus\eta_0\mathbb{B}_{\infty}$. 
The recursively updated sets in Algorithm~\ref{alg:HC} satisfy
\begin{align}
\theta^*\in\Theta_t^M\subseteq \overline{\theta}_t\oplus \eta_{t}\mathbb{B}_{\infty}\subseteq \overline{\theta}_{t-1}\oplus\eta_{t-1}\mathbb{B}_{\infty},~\forall t\geq 0. 
\end{align}
\end{lemma}
\begin{proof}
The proof is analogous to~\cite[Lemma~1]{Koehler2019Adaptive}.
\end{proof}
Overall, the complexity of Alg.~\ref{alg:HC} is typically small compared to the MPC optimization problem~\eqref{eq:RAMPC}.  
The usage of multiple measurements in the moving window approach often allows for a significant reduction in the uncertainty set while keeping the simple hypercube parametrization. 
Similar ideas for parameter estimation using moving window updates are used in~\cite[Rk.~4]{lorenzen2019robust}, \cite{chisci1998block}. 
Set membership updates with polytopes are considered in  \cite{tanaskovic2014adaptive,lorenzen2019robust,Lu2019RAMPC,Koehler2019Adaptive,lopez2018adaptive,bujarbaruah2019adaptive}, while the nonlinear approaches in~\cite{adetola2011robust,guay2015robust,gonccalves2016robust,dahliwal2014set,adetola2014adaptive} typically consider some ellipsoidal sets $\widetilde{\Theta}_t=\{\theta|~\|\theta\|_{\Sigma_t}\leq 1\}$ corresponding to a sublevel set of a Lyapunov function of the estimation scheme. 

\subsubsection*{Incremental Lyapunov function and contraction rate $\rho_\theta$}
A simple class of incremental Lyapunov functions (Ass.~\ref{ass:increm}) is given by $V_{\delta}(x,z)=\|x-z\|_{P(z)}$ and feedback $\kappa(x,z,v)=v+K(x)x-K(z)z$, with $K,P$ nonlinearly parametrized. 
Similar functions are also employed in the numerical examples in~\cite{Robust_TAC_19,Soloperto2019Collision,JK_periodic_automatica,JK_QINF}, which can be computed offline using LMIs and a quasi-LPV parametrization~\cite{JK_QINF}. 
For simplicity, in the following we only consider ellipsoidal tubes $V_{\delta}(x,z)=\|x-z\|_P$ to simplify the design. 
For a given matrix $P\succ 0$, condition~\eqref{eq:increm_a} is directly satisfied with $c_{\delta,l}:=\sqrt{\lambda_{\min}(P)}$, $c_{\delta,u}:=\sqrt{\lambda_{\max}(P)}$.
Conditions~\eqref{eq:increm_d}--\eqref{eq:increm_e} follow from the triangular inequality with $L_{\delta}=0$. 
Condition~\eqref{eq:increm_b} holds with some $\kappa_{\max}\leq \max_{z\in\mathcal{Z}_x}\|K(z)\|/c_{\delta,l}$. 
Condition~\eqref{eq:increm_c} requires the computation of the contraction rate $\rho_{\theta}$ offline, which is similar to the computation of a Lipschitz constant.

The following proposition shows how the computation of $L_{\Delta\Theta,\rho}$ in Prop.~\ref{prop:properties_2} simplifies in this case. 
\begin{proposition}
\label{prop:properties_2_quadratic}
Let Assumption~\ref{ass:model} hold. 
Suppose Assumption~\ref{ass:increm} holdy with a quadratic incremental Lyapunov function $V_{\delta}(x,z)=\|x-z\|_P$. 
For any set $\Delta \Theta=\Delta\eta\mathbb{B}_{\infty}$, $\Delta\eta\geq 0$, and any $\theta^+\in\theta\oplus\Delta\Theta$, Inequality~\eqref{eq:rho_Lipschitz} holds with  $L_{\rho,\theta,\Delta{\Theta}}:=\Delta\eta L_{\mathbb{B},\rho}$, 
\begin{align}
\label{eq:L_B_rho}
L_{\mathbb{B},\rho}:=\max_j \max_{(x,z,v)\in\Psi}\dfrac{\|(G(x,\kappa(x,z,v))-G(z,v))\theta^j\|_P}{\|x-z\|_P},
\end{align}
with $\theta^j\in\text{vert}(\mathbb{B}_{\infty})$, $j=1,\dots,2^p$. 
\end{proposition}
\begin{proof}
The statement directly follows using the contraction property~\eqref{eq:increm_c}, the triangular inequality and linearity in the parameters
\begin{align*}
&\|f_{\theta^+}(x,\kappa(x,z,v))-f_{\theta^+}(z,v)\|_P\\
\leq &\|f_{\theta}(x,\kappa(x,z,v))-f_{\theta}(z,v)\|_P\\
&+\|(G(x,\kappa(x,z,v))-G(z,v))\Delta\theta\|_P\\
\stackrel{\eqref{eq:increm_c}}{\leq} &\rho_{\theta}\|x-z\|_P+\Delta\eta\max_{j}\|(G(x,\kappa(x,z,v))-G(z,v))\Delta\theta^j\|_P\\
\stackrel{\eqref{eq:L_B_rho}}{\leq} &(\rho_{\theta}+\Delta\eta L_{\mathbb{B},\rho})\|x-z\|_P.
\end{align*}
\end{proof}
Again, the complexity of computing~\eqref{eq:L_B_rho} is similar to the computation of a Lipschitz constant ($2^p$ times). 

\subsubsection*{Uncertainty Description}
In the following we provide two formulations for $\tilde{w}$ satisfying Assumption~\ref{ass:w_tilde} with different complexity and conservatism.
\begin{proposition}
\label{prop:design_w_1}
Let Assumption~\ref{ass:model} hold. 
Suppose Assumption~\ref{ass:nominal} holds with hypercubes $\widetilde{\Theta}_t=\eta_t\mathbb{B}_\infty$ and Assumption~\ref{ass:increm} holds with a quadratic incremental Lyapunov function $V_{\delta}(x,z)=\|x-z\|_P$. 
Then the following function satisfies Assumption~\ref{ass:w_tilde}
\begin{subequations}
\label{eq:design_w_1}
\begin{align}
\label{eq:design_w_1_1}
\tilde{w}_{\tilde{\Theta}_t,\mathbb{D}}(z,v):=&\eta_t\max_j \|G(z,v)\theta^j\|_P+\overline{d},\\
\label{eq:design_w_1_2}
\overline{d}:=&\max_{d\in\mathbb{D}}\|d\|_P,\quad 
L_{\tilde{\Theta}_t}:=\eta_t L_{\mathbb{B},\rho},
\end{align}
\end{subequations}
with $\theta^j\in\text{vert}(\mathbb{B}_{\infty})$. 
\end{proposition}
\begin{proof}
Condition~\eqref{eq:w_tilde_a} follows with
\begin{align*}
V_{\delta}(\tilde{z}+d_w,\tilde{z})=\|d_w\|_P\leq \max_{d\in\mathbb{D},\theta\in\eta_t\mathbb{B}_\infty}\|G(z,v)\theta+d\|_P\\
\leq \max_{d\in\mathbb{D}}\|d\|_P+\eta_t\max_j\|G(z,v)\theta^j\|_P\stackrel{\eqref{eq:design_w_1}}{=}\tilde{w}_{\tilde{\Theta}_t,\mathbb{D}}(z,v).
\end{align*}
The Lipschitz bound~\eqref{eq:w_tilde_b} holds with
\begin{align*}
&\tilde{w}_{\tilde{\Theta}_t,\mathbb{D}}(x,\kappa(x,z,v))-\overline{d}\\
=&\eta_t \max_i \|G(x,\kappa(x,z,v))\theta^i\|_P\\
\leq &\eta_t\max_i(\|(G(x,\kappa(x,z,v))-G(z,v))\theta^i\|_P\\
&+\|G(z,v)\theta^i\|_P)\\
\stackrel{\eqref{eq:L_B_rho},\eqref{eq:design_w_1_1}}{\leq} &\eta_t L_{\mathbb{B},\rho}\|x-z\|_P+\tilde{w}_{\widetilde{\Theta}_t,\mathbb{D}}(z,v)-\overline{d}.
\end{align*}
Condition~\eqref{eq:w_tilde_c} is trivially satisfied with equality using the fact that $\tilde{w}_{\widetilde{\Theta}_t,\mathbb{D}}$ is affine in $\eta$ and $\overline{d}$, with $\tilde{w}_{\widetilde{\Theta}_t,\{0\}}(z,v)=\tilde{w}_{\widetilde{\Theta}_t,\mathbb{D}}(z,v)-\overline{d}$.
Similarly, condition~\eqref{eq:w_tilde_d} is satisfied with equality using the fact that $L_{\widetilde{\Theta}}$ is linear in $\eta$. 
Condition~\eqref{eq:w_tilde_e} also holds with equality by definition~\eqref{eq:design_w_1_2}.
\end{proof}
\begin{proposition}
\label{prop:design_w_2}
Let Assumption~\ref{ass:model} hold. 
Suppose Assumption~\ref{ass:nominal} holds with hypercubes $\widetilde{\Theta}_t=\eta_t\mathbb{B}_\infty$ and Assumption~\ref{ass:increm} holds with a quadratic incremental Lyapunov function $V_{\delta}(x,z)=\|x-z\|_P$. 
Then the following function satisfies Assumption~\ref{ass:w_tilde}
\begin{subequations}
\label{eq:design_w_2}
\begin{align}
\label{eq:design_w_2_1}
&\tilde{w}_{\tilde{\Theta}_t,\mathbb{D}}(z,v):=\eta_tc_{\mathbb{B}} \|G(z,v)\|_P+\overline{d},\\
\label{eq:design_w_2_2}
&\overline{d}:=\max_{d\in\mathbb{D}}\|d\|_P,\quad 
L_{\tilde{\Theta}_t}:=\eta_t L_{\mathbb{B}},\quad c_{\mathbb{B}}:=\sqrt{p},\\
\label{eq:design_w_2_3}
&L_{\mathbb{B}}:=c_{\mathbb{B}}\max_{(x,z,v)\in\Psi}\dfrac{\|G(x,\kappa(x,z,v))-G(z,v)\|_P}{\|x-z\|_P}.
\end{align}
\end{subequations}
\end{proposition}
\begin{proof}
First, note that for any $\theta\in\mathbb{B}_{\infty}$, we have $\|G\cdot \theta\|_P\leq \|G\|_P\|\theta\|\leq \sqrt{p}\|G\|_P=:c_{\mathbb{B}}\|G\|_P$.
Thus, satisfaction of condition~\eqref{eq:w_tilde_a} and \eqref{eq:w_tilde_e} directly follows Prop.~\ref{prop:design_w_1}, using the fact the that the formulas in~\eqref{eq:design_w_2} for $\tilde{w},L_{\tilde{\Theta}}$ are conservative over approximations of the formulas in~\eqref{eq:design_w_1}. 
Satisfaction of condition~\eqref{eq:w_tilde_b} follows similar to the proof in Prop.~\ref{prop:design_w_1}. 
As in Prop.~\ref{prop:design_w_1}, Conditions~\eqref{eq:w_tilde_c} and \eqref{eq:w_tilde_d} hold since $\tilde{w}$ is affine  and $L_{\widetilde{\Theta}}$  linear in $\eta$. 
\end{proof}
Prop.~\ref{prop:design_w_1} offers a less conservative design, while the design in Prop.~\ref{prop:design_w_2} has a smaller computational complexity. 
In particular, implementing constraint~\eqref{eq:RAMPC_w} requires $2^{p}/2\cdot N$ nonlinear constraints\footnote{%
Note that due to symmetry only half the $2^p$ vertices $\theta^j$ need to be enumerated to evaluate~\eqref{eq:design_w_1}.
} using the formula in~\eqref{eq:design_w_1}, while the simpler formula in~\eqref{eq:design_w_2} only requires $N$ nonlinear constraints. 
In case of linear dynamics with $V_{\delta}$ polytopic, a similar design can be used which can be implemented using \textit{linear} inequality constraints, compare~\cite{Koehler2019Adaptive}.

\subsubsection*{Terminal ingredients}
The following proposition provides a simple design satisfying all the conditions in Assumption~\ref{ass:term}. 
\begin{proposition}
\label{prop:terminal_design}
Let Assumptions~\ref{ass:model} and \ref{ass:w_tilde} hold. 
Suppose Assumption~\ref{ass:increm} holds with a quadratic incremental Lyapunov function $V_{\delta}(x,z)=\|x-z\|_P$.
Suppose there exists a steady-state $(x_s,u_s)\in\mathcal{Z}$ satisfying $f_{\theta}(x_s,u_s)=x_s$ for all $\theta\in\overline{\theta}_0\oplus\tilde{\Theta}_0$. % and denotes $$.  
Assume further that the following inequality holds
\begin{subequations}
\begin{align}
\label{eq:cond_term}
\rho_{\overline{\theta}_0}+L_{\widetilde{\Theta}_0}+c_{x_s}\tilde{w}_{\tilde{\Theta}_0,\mathbb{D}}(x_s,u_s)\leq 1,\\
\label{eq:c_xs}
c_{x_s}:=\min\{ -h_j(x_s,u_s)/c_j,\delta_{loc}\}.
\end{align}
Then the terminal set
\begin{align}
\label{eq:term_set}
\mathcal{X}_{f}=\{(x,s)\in\mathbb{R}^{n+1}|~(V_{\delta}(x,x_s)+s)\leq c_{x_s}\},
\end{align}
with terminal control law $k_f(x)=\kappa(x,x_s,u_s)$  and $\overline{s}=\delta_{loc}$, $\overline{w}=\infty$ satisfies conditions~\eqref{eq:term_1}--\eqref{eq:term_4} in Assumption~\ref{ass:term}.
Suppose further that there exists a constant $c>0$, such that $\ell(x,u)\leq c\|(x-x_s,u-u_s)\|^2$, $\forall (x,u)\in\mathcal{Z}$. 
Then there exists a constant $\alpha>0$, such that the terminal cost 
\begin{align}
\label{eq:term_cost}
V_f(x):=V^2_{\delta}(x,x_s)\frac{\alpha }{1-(\rho_{\overline{\theta}_0}+L_{\widetilde{\Theta}_0})^2}.
\end{align}
satisfies Assumption~\ref{ass:term}.
\end{subequations}
\end{proposition}
\begin{proof}
\begin{subequations}
First note that satisfaction of~\eqref{eq:cond_term} implies satisfaction of \eqref{eq:cond_term} with $\overline{\theta}_0,\tilde{\Theta}_0$ replaced by $\overline{\theta},\tilde{\Theta}$, using $\overline{\theta}\oplus\widetilde{\Theta}\subseteq\overline{\theta}_0\oplus\widetilde{\Theta}_0$,
\begin{align}
\label{eq:prop_term_1}
\rho_{\overline{\theta}}+L_{\widetilde{\Theta}} \stackrel{\eqref{eq:rho_Lipschitz},\eqref{eq:w_tilde_d},\eqref{eq:w_tilde_e}}{\leq} & \rho_{\overline{\theta}_0}+L_{\widetilde{\Theta}_0},\\
\label{eq:prop_term_2}
 \tilde{w}_{\widetilde{\Theta},\mathbb{D}}(x_s,u_s)\stackrel{\eqref{eq:w_tilde_c}}{\leq}&\tilde{w}_{\widetilde{\Theta}_0,\mathbb{D}}(x_s,u_s),\quad c_{x_s}\geq 0.
\end{align}
Note that for all $(x,s)\in\mathcal{X}_f$ the uncertainty bound $\tilde{w}$ satisfies
\begin{align}
\label{eq:prop_term_3}
\tilde{w}_{\delta,\widetilde{\Theta},\mathbb{D}}(x,k_f(x),s)\stackrel{\eqref{eq:w_tilde_g},\eqref{eq:def_w_delta_tilde}}{\leq} &\tilde{w}_{\widetilde{\Theta},\mathbb{D}}(x_s,u_s)+L_{\widetilde{\Theta}}(s+V_{\delta}(x,x_s)).
\end{align}
Furthermore, the next state satisfies
\begin{align}
\label{eq:prop_term_4}
V_{\delta}(x^+,x_s)\leq &\|x^+-f_{\overline{\theta}}(x,k_f(x))\|_P+\|f_{\overline{\theta}}(x,k_f(x))-x_s\|_P\nonumber\\
\stackrel{\eqref{eq:increm_c},(Ass.~\ref{ass:term}-\ref{cond_term_e})}{\leq}&\rho_{\overline{\theta}}\|x-x_s\|_P+\tilde{s}. 
\end{align}
Thus, the RPI property~\eqref{eq:term_1} follows as for all $(x,s)\in\mathcal{X}_f$:
\begin{align*}
&s^+ + V_{\delta}(x^+,x_s)\\
\stackrel{\eqref{eq:prop_term_4}, (Ass.~\ref{ass:term}-\ref{cond_term_d})}{\leq}& \rho_{\overline{\theta}}(s+\|x-x_s\|_P)+\tilde{w}_{\delta,\widetilde{\Theta},\mathbb{D}}(x,k_f(x),s)\\
\stackrel{\eqref{eq:prop_term_3}}{\leq} &( \rho_{\overline{\theta}}+L_{\widetilde{\Theta}})(s+\|x-x_s\|_P)+\tilde{w}_{\widetilde{\Theta},\mathbb{D}}(x_s,u_s)\\
\stackrel{\eqref{eq:term_set},\eqref{eq:prop_term_1},\eqref{eq:prop_term_2}}{\leq }&( \rho_{\overline{\theta}_0}+L_{\widetilde{\Theta}_0})c_{x_s}+\tilde{w}_{\widetilde{\Theta}_0,\mathbb{D}}(x_s,u_s)\stackrel{\eqref{eq:cond_term}}{\leq} c_{x_s}.
\end{align*}
 Satisfaction of the tightened constraints~\eqref{eq:term_2} follows with
 \begin{align*}
h_j(x,k_f(x))+c_j s\stackrel{\eqref{eq:c_j}}{\leq}& h_j(x_s,u_s)+c_j(V_{\delta}(x,x_s)+s)\\
\stackrel{\eqref{eq:term_set}}{\leq}&  h_j(x_s,u_s)+c_j\cdot c_{x_s}\stackrel{\eqref{eq:c_xs}}{\leq} 0.
\end{align*}
Condition~\eqref{eq:term_3}, \eqref{eq:term_4} hold due to the considered choice of $\overline{s}$, $\overline{w}$ and $V_{\delta}(x,x_s)\leq c_{x_s}\leq \delta_{loc}$ for all $(x,s)\in\mathcal{X}_f$. 
The fact that the choice of terminal cost in~\eqref{eq:term_cost} is valid follows from~\cite{alessandretti2016design}, the bounds~\eqref{eq:increm_a}--\eqref{eq:increm_c} and $\rho_{\overline{\theta}}\leq\rho_{\overline{\theta}_0}+L_{\widetilde{\Theta}_0}<1$.
\end{subequations}
\end{proof}
This design is quite simple, as it only requires the computation of two scalars ($c_{x_s},\alpha$) and the verification of a scalar condition~\eqref{eq:cond_term}.
Condition~\eqref{eq:cond_term} ensures that $\{x|~V_{\delta}(x,x_s)\leq c_{x_s}\}$ is an RPI set. The usage of global incremental bounds $\rho_{\theta}$ instead of properties of some local control Lyapunov function (c.f. design in~\cite{chen1998quasi} or \cite[Prop.~5]{Robust_TAC_19}) can introduce conservatism. 
Additional (possibly less conservative) design procedures can be found in Appendix~\ref{app:term}. 
The requirement of having a fixed steady-state $x_s$ independent of online computed parameters $\theta$ may pose practical limitation in case of dynamic operation.
Some improvements in this direction are discussed in~\cite[App.~B]{Koehler2019Adaptive}, although the issue of setpoint tracking with online changing models seems largely unresolved. 

 \subsubsection*{Overall algorithm}
 The following two algorithms summarize the proposed offline design and the online operation.
\begin{algorithm}[h]
\caption{RAMPC - Offline design}
\label{alg:offline}
 Given model (Ass.~\ref{ass:model}), constraints~\eqref{eq:constraint_definition}, initial set $\overline{\theta}_0,\eta_0 $ (Ass.~\ref{ass:nominal}).
Design $V_{\delta}(x,z)=\|x-z\|_P$, feedback $\kappa$ (e.g. LMIs in~\cite{JK_QINF}). \\
Determine constants  $\rho_{\overline{\theta}_0}$~\eqref{eq:increm_c}, $c_j$~\eqref{eq:c_j} $L_{\mathbb{B},\rho}$~\eqref{eq:L_B_rho}, $\delta_{loc}>0$.\\
Set parameter update gain $\mu>0$~\eqref{eq:mu}. \\
Compute $\overline{d}$~\eqref{eq:design_w_1_2}, $\alpha$~\eqref{eq:term_cost}, $c_{x_s}$ \eqref{eq:c_xs} (possible $L_{\mathbb{B}}$~\eqref{eq:design_w_2_3}).\\
Check if condition~\eqref{eq:cond_term} holds.
\end{algorithm}

The main complexity in the offline design is the choice of a suitable function $V_{\delta},\kappa$, e.g. using LMIs in~\cite{JK_QINF}, while the different constants can be computed similarly to a Lipschitz constant. 
If Algorithm~\ref{alg:offline} is executed successfully (condition~\eqref{eq:cond_term} holds), then Asumptions~\ref{ass:nominal}--\ref{ass:term} hold. 
\begin{algorithm}[h]
\caption{RAMPC - Online}
\label{alg:online}
Execute at each time step $t\in\mathbb{N}$:
\begin{algorithmic}
\State Measure state $x_t$.
\State Update $\eta_t,\overline{\theta}_t$ using Algorithm~\ref{alg:HC}. 
\State Update $\rho_{\overline{\theta}_t}$ using $\rho_{\overline{\theta}_t}=\rho_{\overline{\theta}_0}+(\eta_0-\eta_t)L_{\mathbb{B},\rho}$. 
\State Update $\hat{\theta}_t$ using~\eqref{eq:hat_theta_update}. 
\State Solve MPC optimization problem~\eqref{eq:RAMPC}.
\State Apply control input $u_t=u^*_{0|t}$.
\end{algorithmic}
\end{algorithm}

Note that compared to a nominal MPC scheme, the proposed nonlinear RAMPC scheme additionally requires executing Alg.~\ref{alg:HC} (solving LPs in case of polytopic $\mathbb{D}$) and uses additional constraints in~\eqref{eq:RAMPC} (which are also needed for robust MPC approaches, e.g.~\cite{Robust_TAC_19}). 
A similar design for the special case of linear systems with polytopic tubes can be found in~\cite{Koehler2019Adaptive}.

%% file: Approaches_2.tex
%!TEX root = ./Adaptive_Nonlin.tex
%%%%%%%%%%%%%%%%%%%%%%%%%%%%%%%%%%%%%%%%%%%%%%%%%%%%%%%%%%%%%%%%%%%%%%%%%%%%%%%
\subsection{Special case - Lipschitz approaches} 
\label{sec:special_guay}
In the following, we briefly detail the nonlinear RAMPC approach considered in~\cite{adetola2011robust,guay2015robust,gonccalves2016robust} and demonstrate that it is contained in the presented framework as a special case.
\begin{subequations}
\label{eq:Limon}

 \subsubsection*{Parameter estimation} 
 The estimation procedure in~\cite{gonccalves2016robust} uses a filtered regressor $\omega_{t}\in\mathbb{R}^{n\times p}$ of $G(x_t,u_t)$, a recursively updated identifier matrix $\Sigma_t\in\mathbb{R}^{p\times p}$ and a recursive least squares (RLS) like nominal parameter update $\hat{\theta}_t$. 
The result is a parameter set of the form $\tilde{\Theta}_t=\nu_t\mathbb{B}_2$ with some online computed constant $\nu_t\geq 0$ that satisfies Assumption~\ref{ass:nominal}. 
We would like to point out that this update uses the stability properties of the Lyapunov function $V_{\theta,t}=\|\theta-\hat{\theta}_t\|_{\Sigma_t}^2$ combined with a case distinction to only update the \textit{ball-shaped} parameter set $\overline{\theta}_t\oplus\nu_t\mathbb{B}_2$, if it satisfies condition~\eqref{eq:nominal_param} in Assumption~\ref{ass:nominal}. 

 \subsubsection*{Tube propagation} 
In~\cite{adetola2011robust,guay2015robust,gonccalves2016robust} the tube propagation is done with a \textit{ball-shaped} tube without any stabilization, which is equivalent to considering $V_{\delta}(x,z)=\|x-z\|$, $\kappa(x,z,v)=v$ and $\delta_{loc}=\infty$. 
Define the Lipschitz constants as follows
\begin{align}
\label{eq:Limon_Lipschitz_L}
\mathcal{L}_f:=&\max_{(x,u)\in\mathcal{Z},~(z,u)\in\mathcal{Z}}\frac{\|f(x,u)-f(z,u)\|}{\|x-z\|},\\
\label{eq:Limon_Lipschitz_G}
\mathcal{L}_G:=&\max_{(x,u)\in\mathcal{Z},~(z,u)\in\mathcal{Z}}\frac{\|G(x,u)-G(z,u)\|}{\|x-z\|}.
\end{align}
The contraction rate $\rho_\theta$ satisfies $\rho_{\overline{\theta}}\leq \mathcal{L}_f+\mathcal{L}_G\|\overline{\theta}\|$.
Furthermore, similar to Prop.~\ref{prop:properties_2_quadratic}, for $\Delta\Theta=\Delta\nu\mathbb{B}_2$, inequality~\eqref{eq:rho_Lipschitz} holds with $L_{\rho,\theta,\Delta\Theta}:=\Delta\nu \mathcal{L}_G$.
Correspondingly, we have $\rho_{\overline{\theta}_t}+L_{\widetilde{\Theta}_t}\leq \mathcal{L}_f+\mathcal{L}_G\Pi$, with $\Pi=\overline{\theta}_0+\nu_0$.

\subsubsection*{Uncertainty characterization}
The uncertainty characterization is done analogous to Prop.~\ref{prop:design_w_2} with $\eta_tc_{\mathbb{B}}$ replaced by $\nu_t$ and $P=I_n$, resulting in
\begin{align}
\label{eq:Limon_w}
\tilde{w}_{\widetilde{\Theta}_t,\mathbb{D}}(z,v):=\nu_t\|G(z,v)\|+\overline{d},\\
\overline{d}:=\max_{d\in\mathbb{D}}\|d\|,~L_{\widetilde{\Theta}_t}:=\nu_t \mathcal{L}_{G}.\nonumber
\end{align}
As a result, the tube propagation~\eqref{eq:RAMPC_dyn_s}--\eqref{eq:RAMPC_w} is equivalent to
\begin{align}
\label{eq:Limon_s}
&s_{k+1|t}\geq (L_f+L_G\Pi)s_{k|t}+\|G(\overline{x}_{k|t},\overline{u}_{k|t})\|\nu_t+\overline{d}.
\end{align}
Furthermore, in case the constraint set $\mathcal{Z}$ is polytopic, the tightened constraints~\eqref{eq:RAMPC_con} are equivalent to
\begin{align}
\label{eq:Limon_tighten}
&(\overline{x}_{k|t}\oplus s_{k|t}\mathbb{B}_2,\overline{u}_{k|t})\subseteq\mathcal{Z}.
\end{align}
The constraints~\eqref{eq:Limon_s}--\eqref{eq:Limon_tighten} correspond to the formulation proposed in~\cite{gonccalves2016robust}.
The constraints~\eqref{eq:RAMPC_con_sw} are not considered in~\cite{adetola2011robust,guay2015robust,gonccalves2016robust} ($\overline{w}=\overline{s}=\infty$). 

Thus, we have demonstrated that the formulation in~\cite{gonccalves2016robust} is a special case of the proposed RAMPC framework. 
Furthermore, this also implies that the computational complexity of the MPC optimization problem~\eqref{eq:RAMPC} with the parametrization $\tilde{w}$ considered in Prop.~\ref{prop:design_w_2} is equivalent to the approach in~\cite{gonccalves2016robust}.
However, the constant $\mathcal{L}_f+\mathcal{L}_G\Pi$, which characterizes the dynamics of the tube size $s$ in~\eqref{eq:Limon_s}, is often significantly larger then the constant $\rho_{\overline{\theta}_0}+\eta_0L_{\mathbb{B}}$ used with the formulation in Prop.~\ref{prop:design_w_2}. 
Hence, especially for larger horizons $N$, the simple Lipschitz based formulation~\cite{gonccalves2016robust} can become overly conservative, compare the numerical example in Sec.~\ref{sec:num} and  numerical comparisons in~\cite{Robust_TAC_19,kohler2018novel,IncremStochComparison_Mesbah_19}.

 \subsubsection*{Terminal ingredients} 
 Regarding the terminal ingredients, in~\cite{adetola2011robust,guay2015robust,gonccalves2016robust} a robust control invariant (RCI) set $\mathcal{X}_f$ is designed offline and the following terminal constraint $\overline{x}_{N|t}\oplus s_{N|t}\mathbb{B}_2\subseteq\mathcal{X}_f$ is used. 
This condition can be used to ensure robust recursive feasibility for computationally intractable $\min$--$\max$ RAMPC approaches and it also ensures that applying the open-loop optimal input trajectory is feasible. 
However, this simple terminal constraint does not ensure recursive feasibility for tube-based MPC schemes. 
Furthermore, to the best knowledge of the authors, to this date there exists no suitable design procedure that guarantees the desired properties (Ass.~\ref{ass:term}) for such simple Lipschitz based RAMPC approaches. 
We would like to point out that Prop.~\ref{prop:terminal_design} can only be applied if $\mathcal{L}_{f}+\mathcal{L}_G\Pi<1$, which is quite restrictive. 
Alternative designs for the terminal ingredients that may alleviate this restriction are discussed in Appendix~\ref{app:term}. 
 
\end{subequations}

%% file: Theory_6.tex
%!TEX root = ./Adaptive_Nonlin.tex
%%%%%%%%%%%%%%%%%%%%%%%%%%%%%%%%%%%%%%%%%%%%%%%%%%%%%%%%%%%%%%%%%%%%%%%%%%%%%%%
\subsection{Extensions and open issues} 
\label{sec:main_6}
In the following, we briefly discuss some possible extensions and open issues of the proposed RAMPC framework.

\begin{remark}(Time-varying parameters)
\label{rk:stability_time_varying}
A natural extension of the proposed RAMPC framework is to consider time-varying parameters $\theta^*_{t+1}\in(\theta_t^*\oplus\Omega)\cap\Theta_0$ with $\Omega=\omega\mathbb{B}_\infty$, similar to~\cite{bujarbaruah2019adaptive,dahliwal2014set,tanaskovic2019adaptive}. 
%set membership
In this case the set membership update (Alg.~\ref{alg:HC}) needs to consider the non-falsified set $\Delta_{k|t}=\Delta_k\oplus(t-k)\Omega$. 
%robust
Furthermore, in the tube propagation~\eqref{eq:RAMPC_dyn_s}--\eqref{eq:RAMPC_w} a growing parameter set $\Theta_{k|t}=(\overline{\theta}_t\oplus(\eta_t+k\omega)\mathbb{B}_{\infty})\cap \Theta_0$ needs to be considered as also done in~\cite{tanaskovic2019adaptive}. 
%change rho
Note that in the considered nonlinear tube propagation, the intersection with the initial parameter set $\Theta_0$ typically requires a re-centering of the nominal parameters $\overline{\theta}_{k|t}$, which complicates the analysis of robust recursive feasibility~\cite[Sec.~5.2]{Elisa}. 
%stability
Furthermore, in case of time-varying parameters $\theta^*_t\in\overline{\theta}_t\oplus\tilde{\Theta}_t$, the stability result in Theorem~\ref{thm:stability} changes to finite-gain stability w.r.t. both the additive disturbances $d_t$ and $\sqrt{\|\Delta\theta^*_t\|^2+c\|\Delta\theta^*_{t}\|}$ with the change in parameters $\Delta\theta^*_t:=\theta_{t+1}^*-\theta_t^*$ and a suitable constant $c>0$ depending on $\Theta_0$, compare~\cite[Sec.~5.2]{Elisa} for details.  
This weaker property is due to the fact that the bound~\eqref{eq:intermediate_LMS} does not hold for time-varying parameters. 
Considering this, in~\cite[Prop.~2]{bujarbaruah2019adaptive} (LMS updates with time-varying parameters) a term corresponding to $\|\Delta\theta^*_{t}\|$ seems to be missing.
\end{remark}

\begin{remark}
\label{rk:collision}(Nonlinear constraints)
Using Lipschitz continuity of the constraints $h_j$ (Ass.~\ref{ass:model}) with constants $c_j\geq 0$ can be conservative in case of strongly nonlinear constraints $h_j$ and does not allow for non-smooth collision avoidance constraints. 
At the cost of additional computational complexity, the considered tube formulation can be extended to use more general nonlinear continuity conditions in the constraint tightening (c.f.~\cite[Sec.~IV.C]{Robust_TAC_19}) and even general collision avoidance formulations~\cite{Soloperto2019Collision}.
\end{remark}

\begin{remark}
\label{rk:offset} (Offset free control)
Offsets, as considered in~\cite{bujarbaruah2019adaptive}, are a special case of the considered approach with $G,\rho_\theta,\tilde{w}_{\delta}$ constant, which allows for a simpler implementation with a fixed constraint tightening similar to~\cite{kohler2018novel,IncremStochComparison_Mesbah_19}.
\end{remark}

\begin{remark}
\label{rk:affine_parameters} (Class of nonlinear systems)
The assumed model structure (Ass.~\ref{ass:model}) with additive bounded disturbances $d$ and affine uncertain parameter $\theta$ is one of the main restrictions in the considered approach for general uncertain nonlinear systems. 
For nonlinearly parametrized systems, providing set membership estimates (Ass.~\ref{ass:nominal}) is already challenging~\cite{adetola2014adaptive}. 
The set membership approach in Algorithm~\ref{alg:HC} can, for example, be directly extended to systems of the form $E_\theta x^+=f_\theta(x,u)+d$ with $f_\theta,E_\theta$ affine in $\theta$, which is typical for mechanical systems with an uncertain mass/inertia.
However, it is not yet clear how a simple design for $\tilde{w}$  for such nonlinearly parametrized systems can be obtained. 
Furthermore, considering non-parameteric uncertainty descriptions, e.g. using Lipschitz bounds~\cite{manzano2018robust,bujarbaruah2019semi}, would be an interesting extension. 
\end{remark}

%% file: Example.tex
%!TEX root = ./Adaptive_Nonlin.tex
%%%%%%%%%%%%%%%%%%%%%%%%%%%%%%%%%%%%%%%%%%%%%%%%%%%%%%%%%%%%%%%%%%%%%%%%%%%%%%%
\section{Numerical example}
\label{sec:num}
The following example demonstrates the flexibility of the proposed framework and compares the proposed approach with the state of the art formulation in nonlinear RAMPC~\cite{adetola2011robust,guay2015robust,gonccalves2016robust}. 
Further details regarding the following numerical example can be found in Appendix~\ref{app:example}. 
\subsubsection*{System model}
We consider the following nonlinear discrete-time system
\begin{align}
\label{eq:example_model}
x^+=&x+T_0\begin{pmatrix}
1/2(1+x_1)u-x_2\theta_1+d_1\\
1/2(1-4x_2)u+x_1\theta_2+d_2
\end{pmatrix},
\end{align}
with $x=(x_1,x_2)\in\mathbb{R}^2$, $u\in\mathbb{R}$, $\theta=(\theta_1,\theta_2)\in\mathbb{R}^2$, $d=(d_1,d_2)\in\mathbb{R}^2$, sampling time $T_0=0.05$, which is a modified version of the example in~\cite{pin2009robust}.  
The considered constraint set is $\mathcal{Z}=[-0.1,0.1]^2\times [-2,2]$. 
The disturbance set and initial parametric uncertainty are $\mathbb{D}=0.5\cdot 10^{-4}\cdot \mathbb{B}_{\infty}$ and $\overline{\theta}_0\oplus\eta_0\mathbb{B}_{\infty}=[1,1.02]\times [0.98,1]$, with the true (unknown) parameters $\theta^*=(1,1)$. 
 The control goal is to stabilize the origin $(x,u)=(0,0)$ with the quadratic stage cost $\ell(x,u)=0.1\|x\|^2+u^2$.

\subsubsection*{Offline design}
We compute a function $V_{\delta}(x,z)=\|x-z\|_P$ and feedback $\kappa(x,z,v)=v+K(z,v)(x-z)$ offline using LMIs (c.f. App.~\ref{app:example}, \cite{JK_QINF}).
 For the function $\tilde{w}_\delta$, we consider the three different choices based on Prop.~\ref{prop:design_w_1}, Prop.~\ref{prop:design_w_2} and the approach in~\cite{gonccalves2016robust} (c.f. Sec.~\ref{sec:special_guay}), where we use $\nu_t=c_{\mathbb{B}}\eta_t$ for simplicity. 
We consider the terminal set design in Prop.~\ref{prop:terminal_design}. 
The following table summarizes the numerical values, where $\overline{d}_{\max}/\overline{d}$ denotes by how much the size of the disturbances set $\mathbb{D}$ can be increased before condition~\eqref{eq:cond_term} becomes invalid. 
\begin{center}
\begin{tabular}{c|c|c|c}
Approach& Prop.~\ref{prop:design_w_1}&Prop.~\ref{prop:design_w_2}&\cite{gonccalves2016robust} (Sec.~\ref{sec:special_guay})\\\hline
$\rho_{\overline{\theta}_0}+L_{\widetilde{\Theta}_0}$&0.995&0.994&$1.271$\\
$\overline{d}_{\max}/\overline{d}$&$126\%$&$152\%$&NaN
\end{tabular} 
\end{center}
We would like to point out that the condition~\eqref{eq:cond_term} for the terminal ingredients in Prop.~\ref{prop:terminal_design} can be conservative, which is the main reason why we do not consider a larger parametric uncertainty in the considered example. 
Furthermore, this design is only applicable if $\rho_{\overline{\theta}_0}<1$ (independent of the size of $\widetilde{\Theta}_0,\mathbb{D}$) and is thus typically not applicable to the Lipschitz based approach~\cite{gonccalves2016robust} (c.f. Sec.~\ref{sec:special_guay}), since $\mathcal{L}_f>1$ for most nonlinear systems. 
We would like to point out that the terminal set suggested in~\cite{adetola2011robust,guay2015robust,gonccalves2016robust} is in general not sufficient to ensure robust recursive feasibility except for computationally intractable $\min$--$\max$ MPC approaches. 
One alternative design for the terminal set $\mathcal{X}_f$ can be found in Appendix~\ref{app:term}.

\subsubsection*{Tube propagation}
To compare the conservatism of the different RAMPC formulations, we consider an exemplary open-loop trajectory $x^*_{\cdot|0}$, $u^*_{\cdot|0}$ with initial condition $x_0=(0.1,0.1)$ and compute the tube size $s$ for the Lipschitz based approach~\cite{gonccalves2016robust}, and the formulations in Prop.~\ref{prop:design_w_1} and \ref{prop:design_w_2}, which can be seen in Figure~\ref{fig:tube}. 
In addition,  to demonstrate the effect of the set membership update (Algorithm~\ref{alg:HC}), we show the tube size corresponding to Prop.~\ref{prop:design_w_2} using the updated parameter set $\widetilde{\Theta}_{40}$ after $40$ steps with $\eta_{40}/\eta_0\approx 2.5\% $. 
Note that, due to the difference in shape of the two tube formulations, we display the semi-axes of the ellipses, which are given by the interval  $s_{\cdot|0}\cdot [c_{\delta,l},c_{\delta,u}]$. 
While the computational complexity of the formulation using Prop.~\ref{prop:design_w_2}  and the Lipschitz based approach~\cite{gonccalves2016robust} (c.f. Sec.~\ref{sec:special_guay}) are equivalent, we can see that especially for larger horizons (e.g.  $N> 12$) the proposed approach is clearly less conservative. % and thus superior. 
The formulation in Prop.~\ref{prop:design_w_1} reduces the conservatism (as measured by the final tube size $s_{N|t}$) again by approximately $25\%$ compared to Prop.~\ref{prop:design_w_2}. 
This comes, however, at the cost of increased computational complexity as the function $\tilde{w}_{\delta}$ in  Prop.~\ref{prop:design_w_1} requires $2^p/2=2$ inequality constraints, compared to $1$ inequality constraint for Prop.~\ref{prop:design_w_2}.
Finally, the reduction of conservatism  (evaluated using $s_{N|t}$) due to the set membership estimation is more than $55\%$ (using the same robust tube propagation), while the additional online computational complexity required for the parameter estimation is minimal. 
\begin{figure}[hbtp]
\begin{center}
\includegraphics[width=0.5\textwidth]{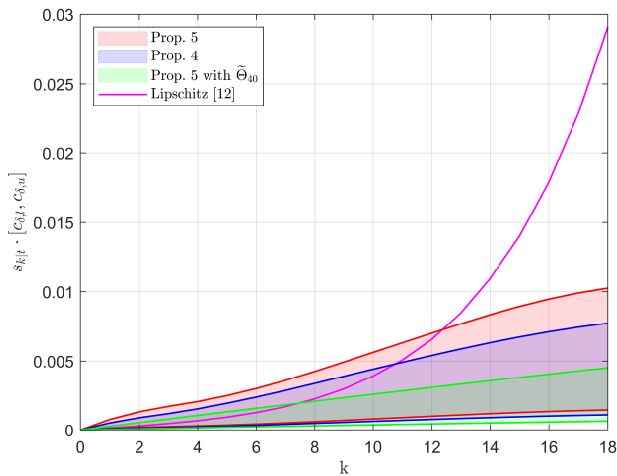}
\end{center}
\caption{%
Spectrum of the elliptic side length of the tube ($s_{\cdot|0}\cdot [c_{\delta,l},c_{\delta,u}]$) for the proposed RAMPC formulations in Prop.~\ref{prop:design_w_1} and \ref{prop:design_w_2}, the Lipschitz based approach~\cite{gonccalves2016robust} (c.f. Sec.~\ref{sec:special_guay}) and Prop.~\ref{prop:design_w_2} with the updated parameter set $\widetilde{\Theta}_{40}$. }
\label{fig:tube}
\end{figure}

\subsubsection*{Region of attraction}
In order to further compare the proposed formulation (for simplicity we only consider Prop.~\ref{prop:design_w_2}) with the Lipschitz based approach~\cite{gonccalves2016robust} (c.f. Sec.~\ref{sec:special_guay}) we compare the region of attraction (ROA) for different horizons $N$. 
As briefly discussed above, the terminal set suggested in~\cite{adetola2011robust,guay2015robust,gonccalves2016robust} is, to the best of our knowledge, not sufficient for ensuring recursive feasibility, and the design proposed in Prop.~\ref{prop:terminal_design} is not applicable to~\cite{adetola2011robust,guay2015robust,gonccalves2016robust} (at least in case that the Lipschitz constant $\mathcal{L}$ is greater than one). Hence, in order to still allow for a meaningful comparison, for the following analysis we consider the terminal set constraint $\mathbb{X}_{N|t}\subseteq\mathcal{X}_f$ for both approaches, with some RPI set $\mathcal{X}_f$ (which ensures open-loop feasibility, but not necessarily closed-loop recursive feasibility).
 In Figure~\ref{fig:ROA}, we see that for small horizons $N$, the ROA is similar, but as the prediction horizon $N$ increases, the ROA increases monotonically for the proposed approach, while after $N\geq 16$ the ROA shrinks drastically for the Lipschitz based approach (and is empty for $N=25$). 
 The following table compares the size of the ROA (all states $x$, for which the optimization problem~\eqref{eq:RAMPC} is feasible)  relative to the constraint set $\mathcal{X}=[-0.1,0.1]^2$. \\
\begin{tabular}{c|c|c|c|c|c|c}
 $\%$ ROA&N=1&N=4&N=9&N=16&N=20&N=25\\\hline
Prop.~\ref{prop:design_w_2}&40&54&67&80&84&88\\
\cite{gonccalves2016robust}&41&55&69&77&27&0
\end{tabular} \\
\begin{figure}[hbtp]
\begin{center}
\includegraphics[width=0.5\textwidth]{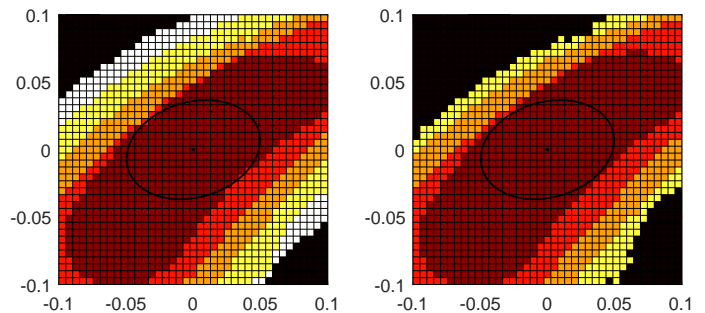}
\end{center}
\caption{Region of attraction with prediction horizons $N\in\{1,4,9,16,25\}$ (dark red, red, orange, yellow, white), infeasible initial conditions (black), terminal region $\mathcal{X}_f$ (black ellipse). Left: Approach using Prop.~\ref{prop:design_w_1}. Right:  Lipschitz approach~\cite{gonccalves2016robust} (Sec.~\ref{sec:special_guay}). }
\label{fig:ROA}
\end{figure}
 
\subsubsection*{Closed-loop performance}
Finally, we focus on the performance in terms of the stage cost $\ell$ and the impact of the LMS point estimate compared to a RMPC approach without parameter adaptation. 
In the given example with $\pm 1\%$ parametric uncertainty, the proposed RAMPC scheme with LMS update improves the performance in terms of the closed-loop stage cost $\sum_{t=0}^{T-1}\ell(x_t,u_t)$ over the first $T=50$ steps by $3.5\%$ compared to a robust MPC implementation without any parameter adaptation.

To summarize, in the considered numerical example we have demonstrated:
(i) reduced conservatism compared to state of the art nonlinear RAMPC approaches~\cite{adetola2011robust,guay2015robust,gonccalves2016robust}, especially  for larger prediction horizons $N$; 
(ii) the degree of freedom in the design of $\tilde{w}_\delta$, determining conservatism and
computational complexity; 
(iii) reduced conservatism and improved performance using online model adaptation (set membership estimation (Alg.~\ref{alg:HC}) and LMS updates~\eqref{eq:hat_theta_update}).

%% file: Sum.tex
%!TEX root = ./Adaptive_Nonlin.tex
%%%%%%%%%%%%%%%%%%%%%%%%%%%%%%%%%%%%%%%%%%%%%%%%%%%%%%%%%%%%%%%%%%%%%%%%%%%%%%%
\section{Conclusion}
\label{sec:sum}
We have presented a framework for nonlinear RAMPC using a simple tube propagation based on incremental Lyapunov functions, thereby extending~\cite{Robust_TAC_19} to online adapted models. 
The proposed approach ensures robust constraint satisfaction and recursive feasibility despite disturbances and parametric uncertainty.
Furthermore, we utilize set membership estimation to reduce the conservatism online, while a LMS point estimate is used to improve performance. 
The framework allows for a flexible trade-off between computational demand and conservatism.
In addition, we have shown that state of the art approaches~\cite{adetola2011robust,guay2015robust,gonccalves2016robust} are a special case of the proposed framework.
We have demonstrated the advantages of the proposed framework with a numerical example.

Extending the proposed framework to a larger class of nonlinear systems (e.g. nonlinearly parametrized) is an open issue.

%% file: Terminal_conservative.tex
%!TEX root = ./Adaptive_Nonlin.tex
%%%%%%%%%%%%%%%%%%%%%%%%%%%%%%%%%%%%%%%%%%%%%%%%%%%%%%%%%%%%%%%%%%%%%%%%%%%%%%%
\subsection{Alternative terminal set}
\label{app:term}    
In the following, we provide a design for the terminal ingredients (Ass.~\ref{ass:term}), as an alternative to Prop.~\ref{prop:terminal_design}. 
While the design in Prop.~\ref{prop:terminal_design} is quite simple, condition~\eqref{eq:cond_term} can be quite restrictive, compare the numerical example. 
In particular, this design does not differentiate between local stability properties of a suitable control Lyapunov function (CLF) and the global incremental stability property using $\rho_\theta$. 
Furthermore, since the design in Prop.~\ref{prop:terminal_design} cannot be applied for systems with $\rho_\theta>1$, it is typically not applicable to Lipschitz based approaches (Sec.~\ref{sec:special_guay}, \cite{adetola2011robust,guay2015robust,gonccalves2016robust}).
In the following proposition we use the fact that $\tilde{w}$ in Prop.~\ref{prop:design_w_1} and \ref{prop:design_w_2} is affine in $\eta$, which allows for simpler bounds.
\begin{proposition}
\label{prop:terminal_design_2}
Suppose the following conditions hold
\begin{subequations}
\label{eq:term2_clf}
\begin{enumerate}[label=\alph*)]
\item There exists a Lipschitz continuous Lyapunov function $V_s(x)$ with Lipschitz continuous feedback $k_f(x)$, a contraction rate $\rho_f\in(0,1)$ and constants $\overline{\gamma}$, $c_{f,l}$, $c_{f,u}>0$, such that the following conditions hold for all $x\in\mathbb{R}^n:$ $V_s(x)\leq \overline{\gamma}$:
\begin{align}
\label{eq:term2_contract}
V_s(f_{\overline{\theta}}(x,k_f(x)))\leq \rho_fV_s(x),\\
\label{eq:term2_bound}
c_{f,l}\|x\|\leq V_s(x)\leq c_{f,u}\|x\|.
\end{align}
\label{cond_a_prop_term_app}
\item Assumption~\ref{ass:model} holds.
\label{cond_b_prop_term_app}
\item Assumption~\ref{ass:nominal} holds with sets $\widetilde{\Theta}_t=\eta_t\mathbb{B}_{\infty}$.
\label{cond_c_prop_term_app}
\item  Assumption~\ref{ass:increm} holds with $V_{\delta}(x,z)=\|x-z\|_P$, $P\succ 0$.
\label{cond_d_prop_term_app}
\item Assumption~\ref{ass:w_tilde} holds with  
\label{cond_e_prop_term_app}
$\tilde{w}_{\widetilde{\Theta},\mathbb{D}}(z,v)=\overline{d}+\eta\tilde{w}_{\mathbb{B}}(z,v)$, $L_{\widetilde{\Theta}}=\eta L_{\mathbb{B}}$, with some $L_{\mathbb{B}}\geq 0$ and  $\tilde{w}_{\mathbb{B}}$ Lipschitz continuous. 
\end{enumerate}
\end{subequations}
Then the following results hold
\begin{subequations}
\label{eq:term2_clf_prop}
\begin{enumerate}
\item There exist constants $\tilde{c}_j\geq 0$, $j=1,\dots,r$, $\tilde{L}_{\mathbb{B}}$, $c_{f,\delta}\geq 0$, such that for any $V_s(x)\leq \overline{\gamma}$, we have: 
\begin{align}
\label{eq:tilde_c_j} 
h_j(x,k_f(x))\leq&  h_j(0,0)+ \tilde{c}_jV_s(x),\\ 
\label{eq:termProp_bound_w}
\tilde{w}_{\widetilde{\Theta},\mathbb{D}}(x,k_f(x))\leq &\overline{d}+\eta(\tilde{w}_{\mathbb{B}}(0,0)+\tilde{L}_{\mathbb{B}}V_s(x)),\\
\label{eq:term2_cont}
V_s(x+d_w)\leq &V_s(x)+c_{f,\delta}V_{\delta}(x+d_w,x).
\end{align}
\item There exist constants $\underline{\rho},\overline{\rho}\geq 0$, such that for any $\overline{\theta}\oplus\eta \mathbb{B}_{\infty}\subseteq\overline{\theta}_0\oplus\eta_0\mathbb{B}_\infty$, we have
\begin{align}
\label{eq:rho_bound}
\underline{\rho}\leq \rho_{\overline{\theta}}+\eta L_{\mathbb{B}}\leq \overline{\rho}.
\end{align}
\end{enumerate}
\end{subequations}
Furthermore, the results in Theorem~\ref{thm:main} remain true if the properties in Assumption~\ref{ass:term} only hold for $\tilde{s}$ satisfying the more restrictive condition:
\begin{align}
\label{eq:s_tilde_more_restrictive}
&\tilde{s}- \underline{\rho}^N\overline{d}\\
\geq&\dfrac{\Delta \eta}{\eta}\left(\frac{\underline{\rho}}{\overline{\rho}}\right)^N\left[(\rho_{\overline{\theta}}+L_{\mathbb{B}}\eta)s+\eta\tilde{w}_{\mathbb{B}}(x,k_f(x))-\sum_{j=1}^{N}\overline{\rho}^j\overline{d}\right],\nonumber
\end{align}
Consider $\overline{w}_{\widetilde{\Theta}}=\overline{d}+\eta\overline{w}_{\mathbb{B}}$ and the terminal set
\begin{align}
\label{eq:term_set_def}
\mathcal{X}_{f,\eta}=\{(x,s)\in\mathbb{R}^{n+1}|~s\in[0,\overline{s}_{f,\eta}],~V_s(x)\leq \gamma_{\eta}\}
\end{align}
with $\gamma_{\eta}:=\gamma_0-\eta \gamma_1$, $\overline{s}_{f,\eta}:=\overline{s}_{f,0}+\eta \overline{s}_{f,1}$ and constants
 $\overline{w}_{\mathbb{B}}$, $\overline{s}_{f,0}$, $\overline{s}_{f,1}$, $\gamma_0$, $\gamma_1\geq 0$.  
Define the polytope $\Omega:=\{(\eta,\Delta\eta\in[0,\eta_0]^2|~\eta-\Delta\eta\geq 0\}$
Suppose the following conditions hold for $(\eta,\Delta\eta)\in\text{vert}(\Omega)$:
\begin{subequations}
\label{eq:termProp_cond_all}
\begin{align}
\label{eq:termProp_cond_constraints}
&h_j(0,0)+c_j \overline{s}_{f,\eta}+\tilde{c}_j\gamma_{\eta} \leq 0,\\
\label{eq:termProp_cond_constraint_s}
&\overline{s}_{f,\eta_0}=\overline{s}_{f,0}+\eta_0 \overline{s}_{f,1}\leq \overline{s}:=\delta_{loc}.\\
\label{eq:termProp_cond_constraint_gamma}
&\gamma_0\leq \overline{\gamma},\\
\label{eq:termProp_cond_RPI_gamma}
&(\rho_f-1)\gamma_\eta+c_{f,\delta}\overline{\rho}^N(\overline{d}+\eta\overline{w}_{\mathbb{B}})\leq \Delta\eta(\gamma_1-c_{f,\delta}\overline{w}_{\mathbb{B}}\sum_{k=0}^{N-1}\overline{\rho}^k)\\
\label{eq:termProp_cond_RPI_s}
&(\overline{\rho}-1)\overline{s}_{f,\eta}+\eta\tilde{w}_{\mathbb{B}}(0,0)+\eta L_{\mathbb{B}}\gamma_{\eta_0}\leq  (\underline{\rho}^N-1)\overline{d},\\
\label{eq:termProp_cond_RPI_s_3}
&0\leq \overline{s}_{f,1}\leq \dfrac{1}{\eta_0}(\underline{\rho}/\overline{\rho})^N\left[\overline{s}_{f,0}-\overline{d}\sum_{j=0}^{N-1}\overline{\rho}^j\right],\\
\label{eq:termProp_cond_w_bar}
&\overline{w}_{\mathbb{B}}-\tilde{w}_{\mathbb{B}}(0,0)\geq L_{\mathbb{B}}\overline{s}_{f,\eta}+\tilde{L}_{\mathbb{B}}\gamma_{\eta}.
\end{align}
\end{subequations}
Then conditions~\eqref{eq:term_1}--\eqref{eq:term_4}, \eqref{eq:term_6} in Assumption~\ref{ass:term} under restriction~\eqref{eq:s_tilde_more_restrictive} are satisfied. 
\end{proposition}
\begin{proof}

\begin{subequations}
\textbf{Part I. } Inequality~\eqref{eq:tilde_c_j}  (similar to  $c_j$ in~\eqref{eq:c_j}) is satisfied with 
\begin{align}
\label{eq:tilde_c_j_compute}
\tilde{c}_j:=\max_{\{x|V_s(x)\leq \overline{\gamma}\}} \dfrac{h_j(x,k_f(x))-h_j(0,0)}{V_s(x)},
\end{align}
 where existence is guaranteed using $k_f,h_j$ Lipschitz continous and the  bound~\eqref{eq:term2_bound} on $V_{s}$. 
Consider the constant 
\begin{align}
\label{eq:tilde_L_theta}
\tilde{L}_{\mathbb{B}}:=\max_{\{x|V_s(x)\leq \overline{\gamma}\}}\dfrac{\tilde{w}_{\mathbb{B}}(x,k_f(x))-\tilde{w}_{\mathbb{B}}(0,0)}{V_s(x)},
\end{align}
 which is finite due to $\tilde{w}_{\mathbb{B}}$ Lipschitz and $V_s$ lower bounded.
Inequality~\eqref{eq:termProp_bound_w} follows by definition. 
Condition~\eqref{eq:term2_cont}  holds with 
\begin{align}
\label{eq:c_f_delta}
c_{f_\delta}:=\max_{\{x|V_s(x)\leq \overline{\gamma}\},d_w\in\mathbb{R}^n}\dfrac{V_s(x+d_w)-V_s(x)}{V_{\delta}(x+d_w,x)},
\end{align}
which is finite since $V_s$ is Lipschitz and $V_\delta$ is lower bounded using~\eqref{eq:increm_a}. 
Regarding inequality~\eqref{eq:rho_bound}, we have $\rho_{\overline{\theta}}+\eta L_{\mathbb{B}}\leq\rho_{\overline{\theta}_0}+\eta_0 L_{\mathbb{B}}:=\overline{\rho}$ due to condition~\eqref{eq:rho_Lipschitz}, \eqref{eq:w_tilde_e} and $\overline{\theta}\in\overline{\theta}_0\oplus(\eta_0-\eta)\mathbb{B}_{\infty}$. 
Similar to Proposition~\ref{prop:terminal_design}, the following lower bound holds using the reverse triangular inequality
\begin{align*}
&\rho_{\overline{\theta}}=\max_{(x,z,v)\in\Psi}\dfrac{\|(G(x,\kappa(x,z,v))-G(z,v))\overline{\theta}\|_P}{\|x-z\|}\\
\geq &\max_{(x,z,v)\in\Psi,\tilde{\theta}\in\eta\mathbb{B}_{\infty}} \dfrac{\|(G(x,\kappa(x,z,v))-G(z,v))\overline{\theta}_0\|_P}{\|x-z\|}\\
&-\dfrac{\|(G(x,\kappa(x,z,v))-G(z,v))\tilde{\theta}\|_P}{\|x-z\|}\\
\stackrel{\eqref{eq:L_B_rho}}{\geq} &\rho_{\overline{\theta}_0}-\eta L_{\mathbb{B},\rho}=:\underline{\rho}.
\end{align*}
\textbf{Part II. }
In order to show that we can consider the more restrictive condition~\eqref{eq:s_tilde_more_restrictive}, it suffices to show that the candidate solution in Theorem~\ref{thm:main} satisfies this inequality.
In particular, noting that in the proof $\tilde{s}_{N|t+1}\equiv \tilde{s}$ and $s^*_{N+1|t}= (\rho_{\overline{\theta}_t}+\eta_t L_{\mathbb{B}})s^*_{N|t}+\eta_t\tilde{w}_{\mathbb{B}}(\overline{x}^*_{N|t},k_f(\overline{x}^*_{N|t}))+\overline{d}$, this is equivalent to showing that the following inequality holds for $k=N$
\begin{align}
\label{eq:term_less_cons}
\tilde{s}_{k|t+1}-\underline{d}_k\geq \dfrac{\Delta\eta_t}{\eta_t} (\underline{\rho}/\overline{\rho})^k \left[s^*_{k+1|t}-\overline{d}_k\right].
\end{align}
with $\underline{d}_k:=\underline{\rho}^j\overline{d}$, $\overline{d}_k:=\sum_{j=0}^k\overline{\rho}^j\overline{d}$. 
We prove~\eqref{eq:term_less_cons} using induction. 
For $k=0$, we have
\begin{align*}
\tilde{s}_{0|t+1}-\overline{d}=s^*_{1|t}-\overline{d}\geq \dfrac{\Delta\eta_t}{\eta_t}[s^*_{1|t}-\overline{d}].
\end{align*}
Induction step:
Suppose~\eqref{eq:term_less_cons} holds from some $k\in\{0,\dots,N-1\}$, then we have 
\begin{align*}
&\tilde{s}_{k+1|t+1}-\underline{d}_{k+1}\\
\stackrel{\eqref{eq:def_tilde_s},\eqref{eq:rho_bound}}\geq& \underline{\rho}\tilde{s}_{k|t+1}-\underline{d}_{k+1}+\Delta\eta_t\tilde{w}_{\mathbb{B}}(\overline{x}^*_{k+1|t},\overline{u}^*_{k+1|t})\\
=&\underline{\rho}(\tilde{s}_{k|t+1}-\underline{d}_k)+\Delta\eta\tilde{w}_{\mathbb{B}}(\overline{x}^*_{k+1|t},\overline{u}^*_{k+1|t})\\
\stackrel{\eqref{eq:term_less_cons}}{\geq} &(\underline{\rho}^{k+1}/\overline{\rho}^k)\dfrac{\Delta\eta_t}{\eta_t}[s^*_{k+1|t}-\overline{d}_k]+\Delta\eta \tilde{w}_{\mathbb{B}}(\overline{x}^*_{k+1|t},\overline{u}^*_{k+1|t})\\
=&(\underline{\rho}/\overline{\rho})^{k+1}\dfrac{\Delta\eta_t}{\eta_t}[\overline{\rho}_ks^*_{k+1|t}-\overline{\rho}\overline{d}_k]+\Delta\eta \tilde{w}_{\mathbb{B}}(\overline{x}^*_{k+1t},\overline{u}^*_{k+1|t})\\
\stackrel{\eqref{eq:RAMPC_dyn_s},\eqref{eq:rho_bound}}{\geq} &(\underline{\rho}/\overline{\rho})^{k+1}\frac{\Delta\eta_t}{\eta_t}[s^*_{k+2|t}-\eta_t\tilde{w}_{\mathbb{B}}(\overline{x}^*_{k+1|t},\overline{u}^*_{k+1|t}) -\overline{d}\\
&-\overline{\rho}\overline{d}_k]
+\Delta\eta \tilde{w}_{\mathbb{B}}(\overline{x}^*_{k+1t},\overline{u}^*_{k+1|t})\\
\geq &(\underline{\rho}/\overline{\rho})^{k+1}\dfrac{\Delta\eta_t}{\eta_t}[s^*_{k+2|t}-\overline{d}_{k+1}]\\
&+(1-(\underline{\rho}/\overline{\rho})^N)\Delta\eta\tilde{w}_{\mathbb{B}}(x^*_{k+1|t},u^*_{k+1|t})\\
\geq &(\underline{\rho}/\overline{\rho})^{k+1}\dfrac{\Delta\eta_t}{\eta_t}[s^*_{k+2|t}-\overline{d}_{k+1}].
\end{align*}
\textbf{Part III. } Assumption~\ref{ass:term}:
Property~\eqref{eq:term_2} holds if 
\begin{align*}
&h_j(x,k_f(x))+c_js\stackrel{\eqref{eq:tilde_c_j},\eqref{eq:term_set_def}}{\leq} h_j(0,0)+\tilde{c}_j\gamma_{\eta}+ c_j\overline{s}_{f,\eta}\\
=&h_j(0,0)+c_j (\overline{s}_{f,0}+\eta \overline{s}_{f,1})+\tilde{c}_j (\gamma_0-\eta\gamma_1)\leq 0,
\end{align*}
for all $\eta\in[0,\eta_0]$
Since this condition is affine in $\eta$, it holds   \textit{iff} it holds for $\eta=0$ and $\eta=\eta_0$ and is thus satisfied due to~\eqref{eq:termProp_cond_constraints}. 
Condition~\eqref{eq:term_4} holds with $s_{f,\eta}\leq s_{f,\eta_0}=\overline{s}\leq \delta_{loc}$~\eqref{eq:termProp_cond_constraint_s}.\\
Condition~\eqref{eq:termProp_cond_constraint_gamma} ensures that the properties~\eqref{eq:term2_clf}--\eqref{eq:term2_clf_prop} can be invoked in the terminal set.\\
The robust positive invariance condition~\eqref{eq:term_1}:
The RPI condition in $x,\gamma$ holds with 
\begin{align*}
&V_s(x^+)\stackrel{\eqref{eq:term2_cont}}{\leq}  V_s(f_{\overline{\theta}}(x,k_f(x)) )+c_{f,\delta}V_{\delta}(x^+,f_{\overline{\theta}}(x,k_f(x)))\\
\stackrel{\eqref{eq:term2_contract}}{\leq}& \rho_fV_s(x)+c_{f,\delta}\tilde{s}\\
\stackrel{\eqref{eq:rho_bound},\eqref{eq:term_set_def}}{\leq}&\rho_f\gamma_{\eta}+c_{f,\delta}\left[(\overline{\rho}^N\overline{w}_{\widetilde{\Theta}}+\Delta\eta\overline{w}_{\mathbb{B}}\sum_{k=0}^{N-1}(\rho_{\overline{\theta}}+\Delta\eta L_{\mathbb{B}})^k\right]\\
\stackrel{\eqref{eq:rho_bound}}{\leq} &\rho\gamma_\eta+c_{f,\delta}\overline{\rho}^N(\eta\overline{w}_{\mathbb{B}}+\overline{d})+\Delta\eta c_{f,\delta}\overline{w}_{\mathbb{B}}\sum_{k=0}^{N-1}\overline{\rho}^k\\
\stackrel{\eqref{eq:termProp_cond_RPI_gamma}}{\leq} &\gamma_{\eta}+\Delta\eta\gamma_1=\gamma_{\eta^+}.
\end{align*}
where the last inequality follows from the fact that~\eqref{eq:termProp_cond_RPI_gamma}  is affine in $(\eta,\Delta \eta)$ and thus it suffices to verify the inequality at the vertices.
Note that the following bound holds
\begin{align}
\label{eq:term_prop_alternative_s_intermediate}
&\rho_{\overline{\theta}} s+\tilde{w}_{\delta,\widetilde{\Theta},\mathbb{D}}(x,k_f(x),s)\\
=&(\rho_{\overline{\theta}}+\eta L_{\mathbb{B}})s+\eta\tilde{w}_{\mathbb{B}}(x,k_f(x))+\overline{d}\nonumber\\
\stackrel{\eqref{eq:termProp_bound_w}}{\leq}& (\rho_{\overline{\theta}}+\eta L_{\mathbb{B}}) {s}+ \eta(\tilde{w}_{\mathbb{B}}(0,0)+L_{\mathbb{B}}V_s(x))+\overline{d}\nonumber\\
\stackrel{\eqref{eq:rho_bound}}{\leq}& \overline{\rho} \overline{s}_{f,\eta}+\overline{d}+\eta\tilde{w}_\mathbb{B}(0,0)+\eta  L_{\mathbb{B}} \gamma_{\eta_0}.\nonumber
\end{align}
First, we derive the following bound 
\begin{align*}
&s^+\leq (\rho_{\overline{\theta}}+\eta L_{\mathbb{B}}) s+\eta\tilde{w}_{\mathbb{B}}(x,k_f(x))+\overline{d}-\tilde{s}\\
\stackrel{\eqref{eq:s_tilde_more_restrictive}}{\leq}&  \left[(\rho_{\overline{\theta}}+\eta L_{\mathbb{B}}) s+\eta\tilde{w}_{\mathbb{B}}(x,k_f(x))\right]
(1-\frac{\Delta\eta}{\eta}(\underline{\rho}/\overline{\rho})^N)\\
&+\overline{d}\left(1-\underline{\rho}^N+\dfrac{\Delta\eta}{\eta}(\underline{\rho}/\overline{\rho})^N\sum_{j=1}^N \overline{\rho}\right).\\
\stackrel{\eqref{eq:term_prop_alternative_s_intermediate}}{\leq}&\left[\overline{\rho}\overline{s}_{f,\eta}+\eta\tilde{w}_{\mathbb{B}}(0,0)+\eta L_{\mathbb{B}}\gamma_{\eta_0}\right]
(1-\frac{\Delta\eta}{\eta}(\underline{\rho}/\overline{\rho})^N)\\
&+\overline{d}\left(1-\underline{\rho}^N+\dfrac{\Delta\eta}{\eta}(\underline{\rho}/\overline{\rho})^N\sum_{j=1}^N \overline{\rho}^j\right).\\
\end{align*}
For $\Delta\eta=0$, the RPI condition then directly follows using
\begin{align}
\label{eq:term_prop_alternative_s_intermediate2}
&\overline{\rho}\overline{s}_{f,\eta}+\eta\tilde{w}_{\mathbb{B}}(0,0)+\eta L_{\mathbb{B}}\gamma_{\eta_0}+\overline{d}(1-\underline{\rho}^N).\nonumber\\
\stackrel{\eqref{eq:termProp_cond_RPI_s}}{\leq} &\overline{s}_{f,\eta},
\end{align}
where we use the fact, that the term is affine in $\eta$ and thus attains its maximum for $\eta\in\{0,\eta_0\}$. 
To show robust positive invariant in case $\Delta\eta\neq 0$, we use the same bound to obtain:
\begin{align*}
&\left[\overline{\rho}\overline{s}_{f,\eta}+\eta\tilde{w}_{\mathbb{B}}(0,0)+\eta L_{\mathbb{B}}\gamma_{0}\right]
(1-\frac{\Delta\eta}{\eta}(\underline{\rho}/\overline{\rho})^N)\\
&+\overline{d}\left(1-\underline{\rho}^N+\dfrac{\Delta\eta}{\eta}(\underline{\rho}/\overline{\rho})^N\sum_{j=1}^N \overline{\rho}\right).\\
\leq&\left[  \overline{s}_{f,\eta}-\overline{d}(1-\overline{\rho}^N) \right] (1-\frac{\Delta\eta}{\eta}(\underline{\rho}/\overline{\rho})^N)\\
&+\overline{d}\left(1-\underline{\rho}^N+\dfrac{\Delta\eta}{\eta}(\underline{\rho}/\overline{\rho})^N\sum_{j=1}^N \overline{\rho}\right)\\
=&\dfrac{\Delta\eta}{\eta}(\underline{\rho}/\overline{\rho})^N\left[\overline{d}(1-\overline{\rho}^N+\sum_{j=1}^N\overline{\rho})-\overline{s}_{f,\eta}  \right]+\overline{s}_{f,\eta}\\
\stackrel{\eqref{eq:termProp_cond_RPI_s_3}}{\leq}& \overline{s}_{f,\eta}-\Delta\eta\overline{s}_{f,1}=\overline{s}_{f,\eta^+}.
\end{align*} 
Condition~\eqref{eq:term_3} follows from
\begin{align*}
&\tilde{w}_{\delta,\widetilde{\Theta},\mathbb{D}}(x,k_f(x),s)=\overline{d}+\eta(\tilde{w}_{\mathbb{B}}(x,k_f(x))+\eta L_{\mathbb{B}}s\\
\stackrel{\eqref{eq:w_tilde_b}\eqref{eq:termProp_bound_w}}{\leq}&\overline{d}+\eta(\tilde{w}_{\mathbb{B}}(0,0)+ L_{\mathbb{B}}\overline{s}_{f,\eta}+ \tilde{L}_{\mathbb{B}}\gamma_{\eta})\\
\stackrel{\eqref{eq:termProp_cond_w_bar}}{\leq} &\overline{d}+\eta \overline{w}_{\mathbb{B}}.
\end{align*}
using the fact that $\gamma_\eta,\overline{s}_{f,\eta}$ are affine in $\eta$ and thus the sum attains the  extremum at $\eta\in\{0,\eta_0\}$. 
\end{subequations}
The parametrization of $\overline{w}_{\widetilde{\Theta}}$  directly ensures satisfaction of~\eqref{eq:term_6}.
\end{proof}
The basic parametrization of the proposed terminal set is similar to the design in~\cite[Prop.~5]{Robust_TAC_19} for robust MPC and can be viewed as a generalization of the nominal design procedure in~\cite{chen1998quasi}. 
The design can then be achieved as follows: 
\begin{enumerate}
\item  Design a standard terminal controller $k_f$ with local Lyapunov function $V_s(x)$, e.g. $k_f(x)=Kx$, $V_s(x)=\|x\|_{P_f}$.
\item  Determine corresponding constants: $c_{f,l},c_{f,u},c_{f,\delta},\overline{\gamma},\rho_f$. 
\item Determine constants with  the following LP with some weighting $\lambda\geq 0$:
\begin{align}
\label{eq:term_design_LP}
\max_{\gamma_0,\gamma_1,\overline{s}_{f,0},\overline{s}_1,\overline{w}_{\mathbb{B}}}\gamma_0-\eta_0\gamma_1+\lambda (s_{f_0}+\eta_0\overline{s}_{f,1})\\
\text{s.t. }\eqref{eq:termProp_cond_all} \text{ holds }\forall (\eta,\Delta\eta)\in\text{vert}(\Omega). \nonumber
\end{align} 
\end{enumerate}
We would like to point out, that in this design the prediction horizon $N$ cannot be changed arbitrarily. 
The resulting terminal set is such, that if the parametric uncertainty decreases $\Delta\eta>0$, the size of the terminal region (in $x$, $\gamma$) increases and the maximal size of the tube $\overline{s}_{f}$ shrinks.
Condition~\eqref{eq:termProp_cond_RPI_s_3} poses a lower bound on the tube size $s$ depending on $\overline{d}$, which corresponds to the constant tube size in case of only additive disturbances ($\eta=0$). 
In the proof, we apply bounds of the form $\gamma_\eta\eta\geq \gamma_{\eta_0}\eta$ to arrive at simple linear expression, which introduces some conservatism.

%% file: Example_details.tex
%!TEX root = ./Adaptive_Nonlin.tex
%%%%%%%%%%%%%%%%%%%%%%%%%%%%%%%%%%%%%%%%%%%%%%%%%%%%%%%%%%%%%%%%%%%%%%%%%%%%%%%
\subsection{Numerical example - additional details} 
\label{app:example}  
In the following, we provide additional details regarding the numerical example in Section~\ref{sec:num}, including the offline design of $V_{\delta}$, $\kappa$. 
\subsubsection*{Offline Computations}
First, we describe the computation of $V_{\delta}$, $\kappa$ satisfying Assumption~\ref{ass:increm}. 
The following procedure is similar to the offline design proposed in~\cite{JK_QINF}, utilizing a quasi-LPV parametrization and LMIs.
The following proposition provides sufficient condition for Assumption~\ref{ass:increm} using conditions on the Jacobian of the dynamics and a suitable parametrization of $V_{\delta}$, $\kappa$. 
\begin{proposition}
\label{prop:Jacobian}
Let Assumption~\ref{ass:model} hold and suppose $f$ and $G$ are twice continuously differentiable. 
Consider the Jacobian matrices $A_{\theta}:\mathcal{Z}\rightarrow\mathbb{R}^{n\times n}$, $B_{\theta}:\mathcal{Z}\rightarrow\mathbb{R}^{n\times m}$ defined as
\begin{align*}
A_{\theta}(z,v):=\left[\dfrac{\partial f+G\theta}{\partial x}\right]_{(z,v)},~ 
B_{\theta}(z,v):=\left[\dfrac{\partial f+G\theta}{\partial u}\right]_{(z,v)}.
\end{align*}
Assume there exists a parametrized continuous feedback $K:\mathcal{Z}\rightarrow\mathbb{R}^{m\times n}$ and a positive definite matrix $P\in\mathbb{R}^{n\times n}$, such that the following inequality holds for all $\theta\in\overline{\theta}_0\oplus\widetilde{\Theta}_0$, and all $(z,v)\in\mathcal{Z}$:
\begin{align}
\label{eq:cond_Jac}
A_{cl,\theta}^\top (z,v) P A_{cl,\theta}(z,v)-\tilde{\rho}_\theta^2 P\leq 0,
\end{align}
with the closed-loop system
\begin{align*}
A_{cl,\theta}(z,v):=A_{\theta}(z,v)+B_{\theta}(z,v)K(z,v),
\end{align*}
and some contraction constant $\tilde{\rho}$. 
Then for any constant $\epsilon>0$, there exists a small enough constant $\delta_{loc}>0$, such that the quadratic incremental Lyapunov function $V_{\delta}(x,z)=\|x-z\|_P$ and the control law $\kappa(x,z,v)=v+K(z,v)\cdot (x-z)$ satisfy Assumption~\ref{ass:increm} with $\rho_{\theta}=\tilde{\rho}_\theta+\epsilon$ and constants $c_{\delta,l}$, $c_{\delta,u}$, $\kappa_{\max}>0$, $L_{\delta}=0$.
\end{proposition}
\begin{proof}
First note, that conditions~\eqref{eq:increm_a},\eqref{eq:increm_b},\eqref{eq:increm_d},\eqref{eq:increm_e} are trivially satisfied for any positive definite matrix $P$ and any bounded feedback $K(z,v)$ with $c_{\delta,l}=\sqrt{\lambda_{\min}(P)}$, $c_{\delta,u}=\sqrt{\lambda_{\max}(P)}$, $\kappa_{\max}=\max_{(z,v)\in\mathcal{Z}}\sigma_{\max}(K(z,v))$, where $\sigma_{\max}$ denotes the maximal singular value. 
Satisfaction of~\eqref{eq:increm_b} with $\rho_\theta$ can be shown by using a first order taylor approximation and bounding the remainder in  a small enough neighbourhood $\delta_{loc}$ using the fact that $f$ and $G$ are twice-continuous differentiable, compare~\cite[Lemma 1]{JK_QINF} for similar arguments.  
\end{proof}

In the following we formulate as set of LMIs to compute the matrices $P$, $K(z,v)$ using condition~\eqref{eq:cond_Jac}.
Given a desired contraction rate $\rho$, the polytopic constraint set $\mathcal{Z}$ defined by $L_{j,x},~L_{j,u}$ and the parameter set $\overline{\theta}_o\oplus\widetilde{\Theta}_0$, we determine $P$, $K$ by solving the semidefinit program (SDP) given in~\eqref{eq:LMI}.
Therein $X\in\mathbb{R}^{n\times n}$ is a constant positive definite matrix and $Y(z,v)\in\mathbb{R}^{m\times n}$ is parametrized as $Y(z,v)=Y_0+\sum_i\vartheta_i(z,v) Y_i$, $\vartheta(z,v)=(v,z_1,z_2,v^2,z_1^2,z_2^2,vz_1,vz_2,z_1z_2)^\top\in\mathbb{R}^9$.
The matrices $P$, $K(z,v)$ are then given by $P=X^{-1}$, $K(z,v)=Y(z,v)P$. 
The constraints~\eqref{eq:LMI_2} ensure the desired contraction in~\eqref{eq:cond_Jac} and hence condition~\eqref{eq:increm_b} using Prop.~\ref{prop:Jacobian}.
The constraints~\eqref{eq:LMI_3} ensure that the constants $c_j$ in~\eqref{eq:c_j} satisfy $c_j\leq 1$. 
The objective~\eqref{eq:LMI_1} minimizes $P$.
Conditions on $\overline{d}$, $L_{\mathbb{B}}$ can also be formulated in terms of LMIs, compare e.g. the numerical example in~\cite{Robust_TAC_19}. 
Note that~\eqref{eq:LMI} as stated needs to be verified for all $(z,v)\in\mathcal{Z}$. 
For the considered low dimensional example, we simply gridded the constraint $\mathcal{Z}$ and then solved the resulting LMIs using  SeDuMi-1.3~\cite{sturm1999using}. 
Alternatively, an approach based on quasi-convexity can be used, compare~\cite[Prop.~1]{JK_QINF}. 
In the numerical example, condition~\eqref{eq:increm_b} is satisfied with $\delta_{loc}=22.81$, $\rho_{\overline{\theta}_0}=0.99$, which is verified numerically (similar to~\cite[Alg.~1]{JK_QINF}). 

\begin{table*}
\small{
\begin{subequations}
\label{eq:LMI}
\begin{align}
\label{eq:LMI_1}
&\min_{X,Y_i}-\log\det(X)\\
\label{eq:LMI_2}
\text{s.t. }& \begin{pmatrix}
\rho^2 X&(A_\theta(z,v)X+B_\theta(z,v)Y(z,v))^\top\\
A_\theta(z,v)X+B_\theta(z,v)Y(z,v)&X
\end{pmatrix}\geq 0,\\
\label{eq:LMI_3}
&\begin{pmatrix}
1&L_{j,x}X+L_{j,u}Y(z,v)\\
(L_{j,x}X+L_{j,u}Y(z,v))^\top &X
\end{pmatrix}\geq 0,\\
&j=1,\dots,q,\quad \forall (z,v)\in\mathcal{Z},\quad \forall \theta\in\overline{\theta}_0\oplus\widetilde{\Theta}_0.
\end{align}
\end{subequations}
}
\end{table*} 
\subsubsection*{Terminal set $\mathcal{X}_f$}
As discussed in Section~\ref{sec:num}, in order to compare the ROA we consider a terminal constraint of the form $\mathbb{X}_{N|t}\subseteq\mathcal{X}_f$, with some RPI set $\mathcal{X}_f$.
We consider the RPI set $\mathcal{X}_f=\{x|~\|x\|^2_{P_f}\leq \gamma\}$, $\gamma=0.63$ and a linear terminal controller $k_f x$, with $k_f,~P_f$ according to~\cite{pin2009robust}. 
Given the ellipsoidal tube $\mathbb{X}_{k|t}=\{x|~\|x_{k|t}-x\|_P\leq s_{k|t}\}$, the constraint $\mathbb{X}_{N|t}\subseteq\mathcal{X}_f$ can be formulated  as
\begin{align}
\|\overline{x}_{N|t}\|_{P_f}+c_f s_{N|t}\leq \gamma,
\end{align}
with the constant $c_f=\sqrt{\lambda_{\max}(P_f,P)}$, where $\lambda_{\max}(A,B)$ denotes the maximal generalized eigenvalue satisfying $Av=\lambda B v$. 
For the Lipschitz-based approach we have $P=I$ and thus $c_f=\sqrt{\lambda_{\max}(P_f)}$.

\subsubsection*{Parameter Estimation}
The successively improving set-membership estimate based on Algorithm~\ref{alg:HC} can be seen in Figure~\ref{fig:ParameterEstimation}.
The hypercube contains the true parameter $\theta^*$ for all iterations and shrinks with the time.
\begin{figure}[H]
	\centering
	\input{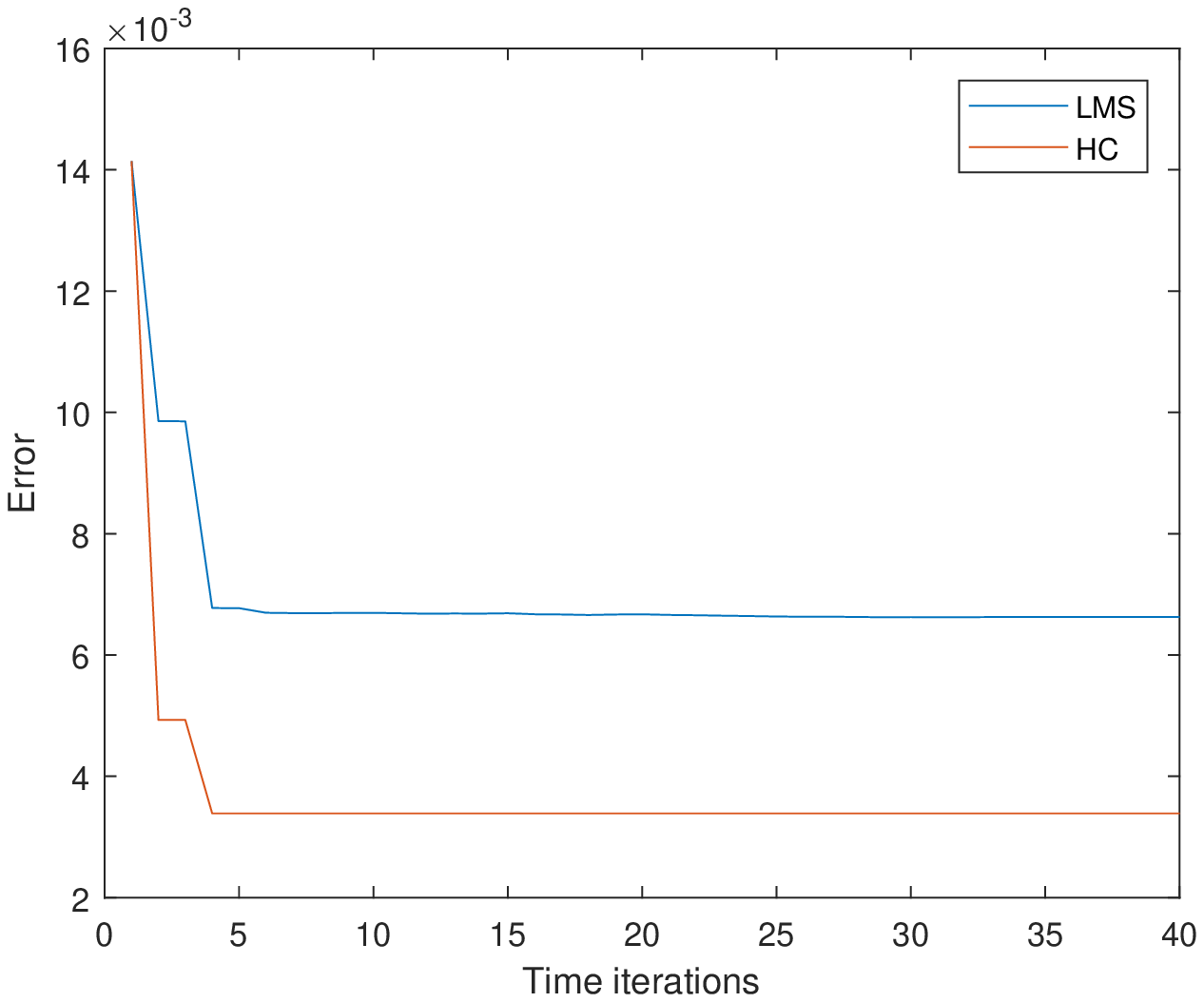}
	\caption{Shrinking parameter sets $\overline{\theta}_t\oplus\eta_t\mathbb{B}_\infty$ at time $t\in\{0,5,10,20,40\}$.}
	\label{fig:ParameterEstimation}
\end{figure} 

\subsubsection*{Open-loop tube}
In Figure~\ref{fig:2D_openloop}, we can see the phase-plot of the open-loop trajectory considered in Figure~\ref{fig:tube}. 
Here, the difference in shape of the tube of the two formulations can be seen more clearly. 
Furthermore, the conservatism of the Lipschitz-based approach is clearly visible.  
\begin{figure}[H]
	\centering
	\includegraphics[width=0.45\textwidth]{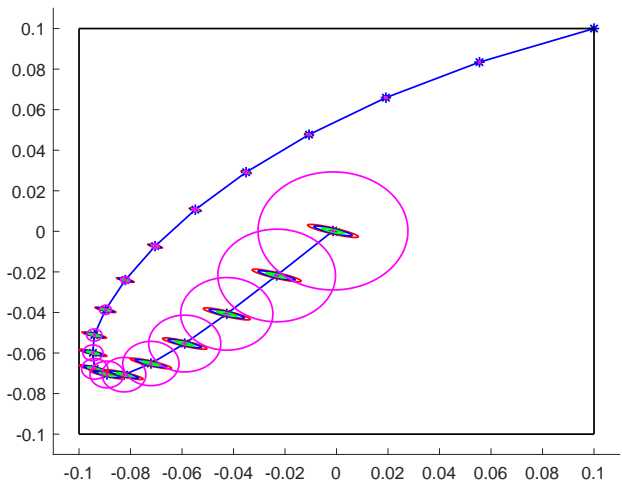}
	\caption{Initially predicted open-loop trajectory $x^*_{\cdot|0}$ (blue), tube $\mathbb{X}_{\cdot|0}^*$ via Prop.~\ref{prop:design_w_2} (red), its adaptive one (green),  tube via Prop.~\ref{prop:design_w_1} (blue),  the Lipschitz based approach~\cite{gonccalves2016robust} (c.f. Sec.~\ref{sec:special_guay}) (magenta), state constraints (black).}
	\label{fig:2D_openloop}
\end{figure}

\subsubsection*{Closed loop performance}
To compare the performance improvement relative to a purely robust MPC formulation (without parameter estimation), we consider a prediction horizon of $N=12$ and repeatedly reinitialize the system at initial points $x_{0}=(0.1,0.1)$ and $x_{0}=(-0.1,-0.1)$ (without reinitializing the parameter updates) and then simulate the system over $T=50$ steps. 
The difference in performance due to the LMS update can be seen in Figure~\ref{fig:LMS_perf}.
There, we can a robust MPC formulation (without parameter adaptation), the proposed RAMPC starting at $t=0$ and also the proposed RAMPC reinitialized at $x_0$ after multiple runs of being reinitialized and thus continuously improving the parameter estimates. 
We can see that the main difference due to improved parameter estimation is already visible at $t=5$. 
Compared to the robust formulation without adaptation, the RAMPC formulation decreases the overshoot. 
Considering the cost $\sum_{t=0}^{T-1}\ell(x_t,u_t)$ with $T=50$, we can see a relative performance improvement of $3.5~\%$ by utilizing the LMS update. 
This performance improvement is expected to increase for larger parametric uncertainty. 
\begin{figure}[hbtp]
\begin{center}
\includegraphics[width=0.45\textwidth]{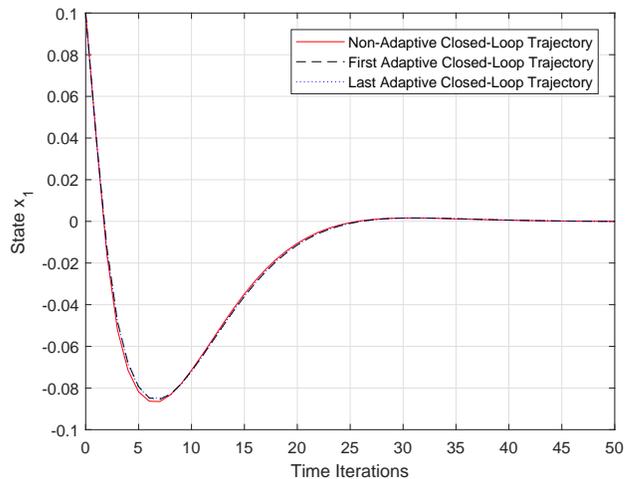}
\end{center}
\caption{Closed-Loop Trajectories with $x_{0}=(0.1,0.1)$:  Robust MPC (without parameter adaptation) (red, solid), RAMPC trajectory (black, dashed) and RAMPC trajectory after multiple simulations (blue, dotted), i.e.,  for $t\in[400,450]$.  }
\label{fig:LMS_perf}
\end{figure}

%% file: Plots/ParameterEstimation.tex
% This file was created by matlab2tikz.
%
%The latest updates can be retrieved from
%  http://www.mathworks.com/matlabcentral/fileexchange/22022-matlab2tikz-matlab2tikz
%where you can also make suggestions and rate matlab2tikz.
%
\begin{tikzpicture}

\begin{axis}[%
width=0.55*4.521in,
height=0.55*3.566in,
at={(0.758in,0.481in)},
scale only axis,
xmin=0.999,
xmax=1.021,
ymin=0.979,
ymax=1.001,
axis background/.style={fill=white},
axis x line*=bottom,
axis y line*=left,
xmajorgrids,
ymajorgrids,
ticklabel style={/pgf/number format/precision=3}
]

\addplot[area legend, line width=1.0pt, draw=black, fill=white!95!lightgray]
table[row sep=crcr] {%
x	y\\
1.02	1\\
1.02	0.98\\
1	0.98\\
1	1\\
}--cycle;
\addlegendentry{t=0}

\addplot[area legend, line width=1.0pt, draw=black, fill=white!75!lightgray]
table[row sep=crcr] {%
x	y\\
1.01109	0.999997833059578\\
1.01109	0.988907833059578\\
1	0.988907833059578\\
1	0.999997833059578\\
}--cycle;
\addlegendentry{t=5}

\addplot[area legend, line width=1.0pt, draw=black, fill=white!75!gray]
table[row sep=crcr] {%
x	y\\
1.00461058628337	0.999996215919351\\
1.00461058628337	0.995386215919351\\
1.00000058628337	0.995386215919351\\
1.00000058628337	0.999996215919351\\
}--cycle;
\addlegendentry{t=10}

\addplot[area legend, line width=1.0pt, draw=black, fill=lightgray]
table[row sep=crcr] {%
x	y\\
1.00048279298899	0.999996215919351\\
1.00048279298899	0.999516215919351\\
1.00000279298899	0.999516215919351\\
1.00000279298899	0.999996215919351\\
}--cycle;
\addlegendentry{t=20}

\addplot[area legend, line width=1.0pt, draw=black, fill=gray]
table[row sep=crcr] {%
x	y\\
1.00048279298899	0.999996215919351\\
1.00048279298899	0.999516215919351\\
1.00000279298899	0.999516215919351\\
1.00000279298899	0.999996215919351\\
}--cycle;
\addlegendentry{t=40}

\end{axis}
\end{tikzpicture}%